\documentclass[journal]{IEEEtran}

\usepackage[nolist,withpage]{acronym}
\usepackage{etoolbox}
\makeatletter
\newif\if@in@acrolist
\AtBeginEnvironment{acronym}{\@in@acrolisttrue}
\newrobustcmd{\LU}[2]{\if@in@acrolist#1\else#2\fi}

\newcommand{\ACF}[1]{{\@in@acrolisttrue\acf{#1}}}
\makeatother
\usepackage{cite}
\usepackage{algorithm}
\usepackage{algorithmicx}
\usepackage{algpseudocode}
\algblock[<block>]{ps}{endps}
\algrenewtext{ps}{\textbf{Parameter Server:}\ }
\algtext*{endps}
\algblock[<block>]{device}{enddevice}
\algrenewtext{device}{\textbf{Device $k$:}\ }
\algtext*{enddevice}
\algblock[<block>]{devices}{enddevices}
\algrenewtext{devices}{\textbf{all devices simultaneously:}\ }
\algtext*{enddevices}
\usepackage[cmex10]{amsmath}
\usepackage{mathtools}
\usepackage{amsfonts}
\usepackage{amssymb}
\usepackage{amsthm}
\DeclareMathOperator*{\argmin}{argmin}

\newtheorem{proposition}{Proposition}
\newtheorem{remark}{Remark}
\newtheorem{assumption}{Assumption}
\newtheorem{lemma}{Lemma}
\newtheorem{corollary}{Corollary}
\DeclareMathAlphabet{\mathppl}{T1}{ppl}{m}{it}
\DeclareMathAlphabet{\mathphv}{T1}{phv}{m}{it}
\DeclareMathAlphabet{\mathpzc}{T1}{pzc}{m}{it}
\newlength{\norlen} 

\newcommand{\Vt}[1]{\mathbf{\lowercase{#1}}}

\newcommand{\vtP}{\Vt{P}}

\newcommand{\vtU}{\Vt{U}}

\newcommand{\vtW}{\Vt{W}}
\newcommand{\vtX}{\Vt{X}}
\newcommand{\vtY}{\Vt{Y}}
\newcommand{\vtZ}{\Vt{Z}}

\newcommand{\vtSigma}{\Vt{\boldsymbol{\sigma}}}

\usepackage{tikz}

\ifCLASSINFOpdf
\else
\fi
\begin{document}

\begin{acronym}[LTE-Advanced]
  \acro{2G}{Second Generation}
  \acro{6G}{Sixth Generation}
  \acro{3-DAP}{3-Dimensional Assignment Problem}
  \acro{AA}{Antenna Array}
  \acro{AC}{Admission Control}
  \acro{AD}{Attack-Decay}
  \acro{ADC}{analog-to-digital converter}
  \acro{ADMM}{alternating direction method of multipliers}
  \acro{ADSL}{Asymmetric Digital Subscriber Line}
  \acro{AHW}{Alternate Hop-and-Wait}
  \acro{AI}{Artificial Intelligence}
  \acro{AirComp}{Over-the-air computation}
  \acro{AMC}{Adaptive Modulation and Coding}
  \acro{AP}{\LU{A}{a}ccess \LU{P}{p}oint}
  \acro{APA}{Adaptive Power Allocation}
  \acro{ARMA}{Autoregressive Moving Average}
  \acro{ARQ}{\LU{A}{a}utomatic \LU{R}{r}epeat \LU{R}{r}equest}
  \acro{ATES}{Adaptive Throughput-based Efficiency-Satisfaction Trade-Off}
  \acro{AWGN}{additive white Gaussian noise}
  \acro{BAA}{\LU{B}{b}roadband \LU{A}{a}nalog \LU{A}{a}ggregation}
  \acro{BB}{Branch and Bound}
  \acro{BCD}{block coordinate descent}
  \acro{BD}{Block Diagonalization}
  \acro{BER}{Bit Error Rate}
  \acro{BF}{Best Fit}
  \acro{BFD}{bidirectional full duplex}
  \acro{BLER}{BLock Error Rate}
  \acro{BPC}{Binary Power Control}
  \acro{BPSK}{Binary Phase-Shift Keying}
  \acro{BRA}{Balanced Random Allocation}
  \acro{BS}{base station}
  \acro{BSUM}{block successive upper-bound minimization}
  \acro{CAP}{Combinatorial Allocation Problem}
  \acro{CAPEX}{Capital Expenditure}
  \acro{CBF}{Coordinated Beamforming}
  \acro{CBR}{Constant Bit Rate}
  \acro{CBS}{Class Based Scheduling}
  \acro{CC}{Congestion Control}
  \acro{CDF}{Cumulative Distribution Function}
  \acro{CDMA}{Code-Division Multiple Access}
  \acro{CE}{\LU{C}{c}hannel \LU{E}{e}stimation}
  \acro{CL}{Closed Loop}
  \acro{CLPC}{Closed Loop Power Control}
  \acro{CML}{centralized machine learning}
  \acro{CNR}{Channel-to-Noise Ratio}
  \acro{CNN}{\LU{C}{c}onvolutional \LU{N}{n}eural \LU{N}{n}etwork}
  \acro{CPA}{Cellular Protection Algorithm}
  \acro{CPICH}{Common Pilot Channel}
  \acro{CoCoA}{\LU{C}{c}ommunication efficient distributed dual \LU{C}{c}oordinate \LU{A}{a}scent}
  \acro{CoMAC}{\LU{C}{c}omputation over \LU{M}{m}ultiple-\LU{A}{a}ccess \LU{C}{c}hannels}
  \acro{CoMP}{Coordinated Multi-Point}
  \acro{CQI}{Channel Quality Indicator}
  \acro{CRM}{Constrained Rate Maximization}
	\acro{CRN}{Cognitive Radio Network}
  \acro{CS}{Coordinated Scheduling}
  \acro{CSI}{\LU{C}{c}hannel \LU{S}{s}tate \LU{I}{i}nformation}
  \acro{CSMA}{\LU{C}{c}arrier \LU{S}{s}ense \LU{M}{m}ultiple \LU{A}{a}ccess}
  \acro{CUE}{Cellular User Equipment}
  \acro{D2D}{device-to-device}
  \acro{DAC}{digital-to-analog converter}
  \acro{DC}{direct current}
  \acro{DCA}{Dynamic Channel Allocation}
  \acro{DE}{Differential Evolution}
  \acro{DFT}{Discrete Fourier Transform}
  \acro{DIST}{Distance}
  \acro{DL}{downlink}
  \acro{DMA}{Double Moving Average}
  \acro{DML}{Distributed Machine Learning}
  \acro{DMRS}{demodulation reference signal}
  \acro{D2DM}{D2D Mode}
  \acro{DMS}{D2D Mode Selection}
  \acro{DNN}{Deep Neural Network}
  \acro{DPC}{Dirty Paper Coding}
  \acro{DRA}{Dynamic Resource Assignment}
  \acro{DSA}{Dynamic Spectrum Access}
  \acro{DSGD}{\LU{D}{d}istributed \LU{S}{s}tochastic \LU{G}{g}radient \LU{D}{d}escent}
  \acro{DSM}{Delay-based Satisfaction Maximization}
  \acro{ECC}{Electronic Communications Committee}
  \acro{EFLC}{Error Feedback Based Load Control}
  \acro{EI}{Efficiency Indicator}
  \acro{eNB}{Evolved Node B}
  \acro{EPA}{Equal Power Allocation}
  \acro{EPC}{Evolved Packet Core}
  \acro{EPS}{Evolved Packet System}
  \acro{E-UTRAN}{Evolved Universal Terrestrial Radio Access Network}
  \acro{ES}{Exhaustive Search}
  \acro{FD}{\LU{F}{f}ederated \LU{D}{d}istillation}
  \acro{FDD}{frequency division duplex}
  \acro{FDM}{Frequency Division Multiplexing}
  \acro{FDMA}{\LU{F}{f}requency \LU{D}{d}ivision \LU{M}{m}ultiple \LU{A}{a}ccess}
  \acro{FedAvg}{\LU{F}{f}ederated \LU{A}{a}veraging}
  \acro{FER}{Frame Erasure Rate}
  \acro{FF}{Fast Fading}
  \acro{FL}{Federated Learning}
  \acro{FSB}{Fixed Switched Beamforming}
  \acro{FST}{Fixed SNR Target}
  \acro{FTP}{File Transfer Protocol}
  \acro{GA}{Genetic Algorithm}
  \acro{GBR}{Guaranteed Bit Rate}
  \acro{GLR}{Gain to Leakage Ratio}
  \acro{GOS}{Generated Orthogonal Sequence}
  \acro{GPL}{GNU General Public License}
  \acro{GRP}{Grouping}
  \acro{HARQ}{Hybrid Automatic Repeat Request}
  \acro{HD}{half-duplex}
  \acro{HMS}{Harmonic Mode Selection}
  \acro{HOL}{Head Of Line}
  \acro{HSDPA}{High-Speed Downlink Packet Access}
  \acro{HSPA}{High Speed Packet Access}
  \acro{HTTP}{HyperText Transfer Protocol}
  \acro{ICMP}{Internet Control Message Protocol}
  \acro{ICI}{Intercell Interference}
  \acro{ID}{Identification}
  \acro{IETF}{Internet Engineering Task Force}
  \acro{ILP}{Integer Linear Program}
  \acro{JRAPAP}{Joint RB Assignment and Power Allocation Problem}
  \acro{UID}{Unique Identification}
  \acro{IID}{\LU{I}{i}ndependent and \LU{I}{i}dentically \LU{D}{d}istributed}
  \acro{IIR}{Infinite Impulse Response}
  \acro{ILP}{Integer Linear Problem}
  \acro{IMT}{International Mobile Telecommunications}
  \acro{INV}{Inverted Norm-based Grouping}
  \acro{IoT}{Internet of Things}
  \acro{IP}{Integer Programming}
  \acro{IPv6}{Internet Protocol Version 6}
  \acro{ISD}{Inter-Site Distance}
  \acro{ISI}{Inter Symbol Interference}
  \acro{ITU}{International Telecommunication Union}
  \acro{JAFM}{joint assignment and fairness maximization}
  \acro{JAFMA}{joint assignment and fairness maximization algorithm}
  \acro{JOAS}{Joint Opportunistic Assignment and Scheduling}
  \acro{JOS}{Joint Opportunistic Scheduling}
  \acro{JP}{Joint Processing}
	\acro{JS}{Jump-Stay}
  \acro{KKT}{Karush-Kuhn-Tucker}
  \acro{L3}{Layer-3}
  \acro{LAC}{Link Admission Control}
  \acro{LA}{Link Adaptation}
  \acro{LC}{Load Control}
  \acro{LDC}{\LU{L}{l}earning-\LU{D}{d}riven \LU{C}{c}ommunication}
  \acro{LOS}{line of sight}
  \acro{LP}{Linear Programming}
  \acro{LTE}{Long Term Evolution}
	\acro{LTE-A}{\ac{LTE}-Advanced}
  \acro{LTE-Advanced}{Long Term Evolution Advanced}
  \acro{M2M}{Machine-to-Machine}
  \acro{MAC}{multiple access channel}
  \acro{MANET}{Mobile Ad hoc Network}
  \acro{MC}{Modular Clock}
  \acro{MCS}{Modulation and Coding Scheme}
  \acro{MDB}{Measured Delay Based}
  \acro{MDI}{Minimum D2D Interference}
  \acro{MF}{Matched Filter}
  \acro{MG}{Maximum Gain}
  \acro{MH}{Multi-Hop}
  \acro{MIMO}{\LU{M}{m}ultiple \LU{I}{i}nput \LU{M}{m}ultiple \LU{O}{o}utput}
  \acro{MINLP}{mixed integer nonlinear programming}
  \acro{MIP}{Mixed Integer Programming}
  \acro{MISO}{multiple input single output}
  \acro{ML}{Machine Learning}
  \acro{MLWDF}{Modified Largest Weighted Delay First}
  \acro{MME}{Mobility Management Entity}
  \acro{MMSE}{minimum mean squared error}
  \acro{MOS}{Mean Opinion Score}
  \acro{MPF}{Multicarrier Proportional Fair}
  \acro{MRA}{Maximum Rate Allocation}
  \acro{MR}{Maximum Rate}
  \acro{MRC}{Maximum Ratio Combining}
  \acro{MRT}{Maximum Ratio Transmission}
  \acro{MRUS}{Maximum Rate with User Satisfaction}
  \acro{MS}{Mode Selection}
  \acro{MSE}{\LU{M}{m}ean \LU{S}{s}quared \LU{E}{e}rror}
  \acro{MSI}{Multi-Stream Interference}
  \acro{MTC}{Machine-Type Communication}
  \acro{MTSI}{Multimedia Telephony Services over IMS}
  \acro{MTSM}{Modified Throughput-based Satisfaction Maximization}
  \acro{MU-MIMO}{Multi-User Multiple Input Multiple Output}
  \acro{MU}{Multi-User}
  \acro{NAS}{Non-Access Stratum}
  \acro{NB}{Node B}
	\acro{NCL}{Neighbor Cell List}
  \acro{NLP}{Nonlinear Programming}
  \acro{NLOS}{non-line of sight}
  \acro{NMSE}{Normalized Mean Square Error}
  \acro{NN}{Neural Network}
  \acro{NOMA}{\LU{N}{n}on-\LU{O}{o}rthogonal \LU{M}{m}ultiple \LU{A}{a}ccess}
  \acro{NORM}{Normalized Projection-based Grouping}
  \acro{NP}{non-polynomial time}
  \acro{NRT}{Non-Real Time}
  \acro{NSPS}{National Security and Public Safety Services}
  \acro{O2I}{Outdoor to Indoor}
  \acro{OFDMA}{\LU{O}{o}rthogonal \LU{F}{f}requency \LU{D}{d}ivision \LU{M}{m}ultiple \LU{A}{a}ccess}
  \acro{OFDM}{Orthogonal Frequency Division Multiplexing}
  \acro{OFPC}{Open Loop with Fractional Path Loss Compensation}
	\acro{O2I}{Outdoor-to-Indoor}
  \acro{OL}{Open Loop}
  \acro{OLPC}{Open-Loop Power Control}
  \acro{OL-PC}{Open-Loop Power Control}
  \acro{OPEX}{Operational Expenditure}
  \acro{ORB}{Orthogonal Random Beamforming}
  \acro{JO-PF}{Joint Opportunistic Proportional Fair}
  \acro{OSI}{Open Systems Interconnection}
  \acro{PAIR}{D2D Pair Gain-based Grouping}
  \acro{PAPR}{Peak-to-Average Power Ratio}
  \acro{P2P}{Peer-to-Peer}
  \acro{PC}{Power Control}
  \acro{PCI}{Physical Cell ID}
  \acro{PDCCH}{physical downlink control channel}
  \acro{PDD}{penalty dual decomposition}
  \acro{PDF}{Probability Density Function}
  \acro{PER}{Packet Error Rate}
  \acro{PF}{Proportional Fair}
  \acro{P-GW}{Packet Data Network Gateway}
  \acro{PL}{Pathloss}
  \acro{PRB}{Physical Resource Block}
  \acro{PROJ}{Projection-based Grouping}
  \acro{ProSe}{Proximity Services}
  \acro{PS}{\LU{P}{p}arameter \LU{S}{s}erver}
  \acro{PSO}{Particle Swarm Optimization}
  \acro{PUCCH}{physical uplink control channel}
  \acro{PZF}{Projected Zero-Forcing}
  \acro{QAM}{Quadrature Amplitude Modulation}
  \acro{QoS}{quality of service}
  \acro{QPSK}{Quadri-Phase Shift Keying}
  \acro{RAISES}{Reallocation-based Assignment for Improved Spectral Efficiency and Satisfaction}
  \acro{RAN}{Radio Access Network}
  \acro{RA}{Resource Allocation}
  \acro{RAT}{Radio Access Technology}
  \acro{RATE}{Rate-based}
  \acro{RB}{resource block}
  \acro{RBG}{Resource Block Group}
  \acro{REF}{Reference Grouping}
  \acro{RF}{radio frequency}
  \acro{RLC}{Radio Link Control}
  \acro{RM}{Rate Maximization}
  \acro{RNC}{Radio Network Controller}
  \acro{RND}{Random Grouping}
  \acro{RRA}{Radio Resource Allocation}
  \acro{RRM}{\LU{R}{r}adio \LU{R}{r}esource \LU{M}{m}anagement}
  \acro{RSCP}{Received Signal Code Power}
  \acro{RSRP}{reference signal receive power}
  \acro{RSRQ}{Reference Signal Receive Quality}
  \acro{RR}{Round Robin}
  \acro{RRC}{Radio Resource Control}
  \acro{RSSI}{received signal strength indicator}
  \acro{RT}{Real Time}
  \acro{RU}{Resource Unit}
  \acro{RUNE}{RUdimentary Network Emulator}
  \acro{RV}{Random Variable}
  \acro{SAC}{Session Admission Control}
  \acro{SCM}{Spatial Channel Model}
  \acro{SC-FDMA}{Single Carrier - Frequency Division Multiple Access}
  \acro{SD}{Soft Dropping}
  \acro{S-D}{Source-Destination}
  \acro{SDPC}{Soft Dropping Power Control}
  \acro{SDMA}{Space-Division Multiple Access}
  \acro{SDR}{semidefinite relaxation}
  \acro{SDP}{semidefinite programming}
  \acro{SER}{Symbol Error Rate}
  \acro{SES}{Simple Exponential Smoothing}
  \acro{S-GW}{Serving Gateway}
  \acro{SGD}{\LU{S}{s}tochastic \LU{G}{g}radient \LU{D}{d}escent}  
  \acro{SINR}{signal-to-interference-plus-noise ratio}
  \acro{SI}{self-interference}
  \acro{SIP}{Session Initiation Protocol}
  \acro{SISO}{\LU{S}{s}ingle \LU{I}{i}nput \LU{S}{s}ingle \LU{O}{o}utput}
  \acro{SIMO}{Single Input Multiple Output}
  \acro{SIR}{Signal to Interference Ratio}
  \acro{SLNR}{Signal-to-Leakage-plus-Noise Ratio}
  \acro{SMA}{Simple Moving Average}
  \acro{SNR}{\LU{S}{s}ignal-to-\LU{N}{n}oise \LU{R}{r}atio}
  \acro{SORA}{Satisfaction Oriented Resource Allocation}
  \acro{SORA-NRT}{Satisfaction-Oriented Resource Allocation for Non-Real Time Services}
  \acro{SORA-RT}{Satisfaction-Oriented Resource Allocation for Real Time Services}
  \acro{SPF}{Single-Carrier Proportional Fair}
  \acro{SRA}{Sequential Removal Algorithm}
  \acro{SRS}{sounding reference signal}
  \acro{SU-MIMO}{Single-User Multiple Input Multiple Output}
  \acro{SU}{Single-User}
  \acro{SVD}{Singular Value Decomposition}
  \acro{SVM}{\LU{S}{s}upport \LU{V}{v}ector \LU{M}{m}achine}
  \acro{TCP}{Transmission Control Protocol}
  \acro{TDD}{time division duplex}
  \acro{TDMA}{\LU{T}{t}ime \LU{D}{d}ivision \LU{M}{m}ultiple \LU{A}{a}ccess}
  \acro{TNFD}{three node full duplex}
  \acro{TETRA}{Terrestrial Trunked Radio}
  \acro{TP}{Transmit Power}
  \acro{TPC}{Transmit Power Control}
  \acro{TTI}{transmission time interval}
  \acro{TTR}{Time-To-Rendezvous}
  \acro{TSM}{Throughput-based Satisfaction Maximization}
  \acro{TU}{Typical Urban}
  \acro{UE}{\LU{U}{u}ser \LU{E}{e}quipment}
  \acro{UEPS}{Urgency and Efficiency-based Packet Scheduling}
  \acro{UL}{uplink}
  \acro{UMTS}{Universal Mobile Telecommunications System}
  \acro{URI}{Uniform Resource Identifier}
  \acro{URM}{Unconstrained Rate Maximization}
  \acro{VR}{Virtual Resource}
  \acro{VoIP}{Voice over IP}
  \acro{WAN}{Wireless Access Network}
  \acro{WCDMA}{Wideband Code Division Multiple Access}
  \acro{WF}{Water-filling}
  \acro{WiMAX}{Worldwide Interoperability for Microwave Access}
  \acro{WINNER}{Wireless World Initiative New Radio}
  \acro{WLAN}{Wireless Local Area Network}
  \acro{WMMSE}{weighted minimum mean square error}
  \acro{WMPF}{Weighted Multicarrier Proportional Fair}
  \acro{WPF}{Weighted Proportional Fair}
  \acro{WSN}{Wireless Sensor Network}
  \acro{WWW}{World Wide Web}
  \acro{XIXO}{(Single or Multiple) Input (Single or Multiple) Output}
  \acro{ZF}{Zero-Forcing}
  \acro{ZMCSCG}{Zero Mean Circularly Symmetric Complex Gaussian}
\end{acronym}

\title{Over-the-Air Federated Learning with Retransmissions (Extended Version)}

\author{Henrik~Hellström,~\IEEEmembership{Student Member,~IEEE,}
        Viktoria~Fodor,~\IEEEmembership{Member,~IEEE,}
        and~Carlo~Fischione,~\IEEEmembership{Senior~Member,~IEEE}
\thanks{All authors were with the School
of Electrical Engineering and Computer Science, KTH - Royal Intitute of Technology, Stockholm,
Sweden. E-mails: (hhells@kth.se, vjfodor@kth.se, carlofi@kth.se).}
\thanks{This work has been submitted to the IEEE for possible publication.  Copyright may be transferred without notice, after which this version may no longer be accessible.}}

\maketitle

\begin{abstract}
Motivated by increasing computational capabilities of wireless devices, as well as unprecedented levels of user- and device-generated data, new distributed machine learning (ML) methods have emerged. In the wireless community, Federated Learning (FL) is of particular interest due to its communication efficiency and its ability to deal with the problem of non-IID data. FL training can be accelerated by a wireless communication method called Over-the-Air Computation (AirComp) which harnesses the interference of simultaneous uplink transmissions to efficiently aggregate model updates. However, since AirComp utilizes analog communication, it introduces inevitable estimation errors. In this paper, we study the impact of such estimation errors on the convergence of FL and propose retransmissions as a method to improve FL convergence over resource-constrained wireless networks. First, we derive the optimal AirComp power control scheme with retransmissions over static channels. Then, we investigate the performance of Over-the-Air FL with retransmissions and find two upper bounds on the FL loss function. Finally, we propose a heuristic for selecting the optimal number of retransmissions, which can be calculated before training the ML model. Numerical results demonstrate that the introduction of retransmissions can lead to improved ML performance, without incurring extra costs in terms of communication or computation. Additionally, we provide simulation results on our heuristic which indicate that it can correctly identify the optimal number of retransmissions for different wireless network setups and machine learning problems.

\end{abstract}


%
\IEEEpeerreviewmaketitle


\section{Introduction}
%
%
%
%
\IEEEPARstart{T}{he} data collection rate in wireless devices is growing at an exceptional speed due to the increasing adoption of smartphones, tablets, and \ac{IoT} devices \cite{ericsson2020, pewresearch2015}. These devices are expected to provide a broad range of \ac{AI} services in \ac{6G} networks, such as predictive healthcare \cite{jagannathan2020predictive}, search-and-rescue drones \cite{brik2020federated}, and environmental monitoring \cite{knoll2019large}. As a consequence, new distributed machine learning methods, such as \ac{FL}, have become essential to enable privacy-preserving and communication-efficient model training \cite{mcmahan2017communication}. A recent survey on open problems of \ac{FL}, argues that communication is often a primary bottleneck for \ac{FL} because wireless links operate at low rates that can be both expensive and unreliable \cite{kairouz2019advances}. Communication-efficient \ac{FL} is investigated thoroughly in \cite{konevcny2017federated}, where various compression techniques such as quantization, random rotation, and subsampling are evaluated. In \cite{hellstrom2020wireless, zhu2020toward, amiri2020machine, chen2020joint} it is established that new wireless methods can greatly improve the communication efficiency of edge \ac{AI}.

A relatively unexplored approach to wireless communication, called \ac{AirComp} (also known as computation over multiple-access channels), has recently been adapted to support \ac{ML} services \cite{nazer2007computation, yang2020federated}. \ac{AirComp} is an analog communication scheme that orders its users to communicate simultaneously over the same frequency band, thereby promoting interference. This interference is leveraged to compute a function of the transmitted messages by utilizing the superposition property of the wireless channel \cite{goldenbaum2013harnessing}. By appropriately precoding the transmitted signals, many functions can be calculated over the air, for instance, the arithmetic mean, product, and min functions \cite{abari2016over}. In \ac{FL}, the central server is interested in collecting the arithmetic mean of model updates from the participating devices, and therefore \ac{AirComp} is a suitable communication solution \cite{zhu2020toward}.

Compared to conventional point-to-point digital communications, \ac{AirComp} is very attractive from a communication-efficiency standpoint, with throughput gains approximately proportional to the number of users \cite{nazer2007computation}. The reason for this drastic improvement is that the entire wireless spectrum can be utilized concurrently by all devices, rather than dividing it and allocating a smaller resource block to each device. Additionally, \ac{AirComp} obfuscates the participating users since the central server directly receives the arithmetic mean rather than the individual model updates, thereby enhancing privacy \cite{hasirciouglu2021private}.

Currently, \ac{AirComp} is reliant upon specialized hardware and fine synchronization that might be difficult to achieve in practice \cite{goldenbaum2013robust}. Additionally, \ac{AirComp} is unable to guarantee perfect reconstruction of the transmitted messages at the receiver. Shannon's "fundamental theorem for a discrete channel with noise" establishes that for any degree of noise contamination, it is possible to communicate discrete data with an arbitrarily small frequency of errors \cite{shannon1948mathematical}. However, to achieve a non-zero communication rate, redundant information must be transmitted in the form of a code \cite{bose1960class}. Since the information transmitted in \ac{AirComp} is not discrete, existing codes do not appear to be applicable. Instead, \ac{AirComp} settles for estimating the desired function as closely as possible, while retaining some non-zero estimation error \cite{cao2020optimized}. In \cite{sery2020analog}, it is proved that these errors reduce the convergence rate of \ac{FL}, thereby requiring additional communication rounds to reach an optimum.

In the current \ac{AirComp} literature, the main way of reducing the estimation error is to optimize the transmission powers. In \cite{cao2020optimized} and \cite{liu2020over}, the authors propose a closed-form power control scheme that minimizes the \ac{MSE} between the received signal and the true sum of the sources' messages. They consider single-antenna devices with a peak transmission power constraint. A similar problem is solved in \cite{zang2020over}, where a sum-power constraint is used to preserve the long-term energy consumption of the participating devices. For the case of multiple antennas, no closed-form power control scheme has been found, but \cite{yang2020federated} develops a strong heuristic by using a difference-of-convex-functions representation of the problem. In \cite{li2019wirelessly}, the multi-antenna problem is coupled with wireless power transfer to improve the battery life of participating \ac{IoT} devices. To further improve the power control, \cite{zhang2021gradient} proposes a gradient-statistics aware scheme that learns statistical properties of the model updates to improve the \ac{AirComp} estimation error. Similarly, \cite{fan2021temporal} learns the temporal structure of gradient sparsity to develop a Bayesian prior that improves the estimation. Finally, several works have incorporated intelligent reflective surfaces with \ac{AirComp} to reach substantially lower estimation errors \cite{jiang2019over, wang2020wireless, wang2021federated}.

As a general pattern, none of these works offer avenues to trade-off communication resources for improved estimation. In digital communications, such communication-estimation trade-offs are the main way to reduce errors. For instance, it is standard to adaptively control the modulation order and coding rate to compensate for poor channels \cite{goldsmith1998adaptive}. Unfortunately, none of these approaches are directly compatible with \ac{AirComp} since the communication is analog. In this paper, we take a first step towards enabling this communication-estimation tradeoff for Over-the-Air federated learning with a system we call AirReComp. The contributions of this paper are summarized as follows.
\begin{itemize}
    \item This is the first work that allows communication resources to be traded for improved estimation in Over-the-Air \ac{FL};
    \item A power control scheme for AirReComp is proposed. The proposed scheme is proven to be globally optimal in terms of \ac{MSE} between the estimated and desired function. 
    \item Upper bounds on the \ac{FL} loss function are derived for single-epoch Lipschitz-smooth functions, both for the strongly-convex and convex case.
    \item Using the bound for convex functions, we develop a heuristic for selecting the number of retransmissions that minimize the \ac{FL} loss given a limited energy budget;
    \item To further support the feasibility of AirReComp under non-convex functions, we provide numerical results with \acp{DNN}. These results suggest that AirReComp can outperform state-of-the-art Over-the-Air \ac{FL} in terms of classification accuracy, without increasing the total cost of training.
\end{itemize}
The remainder of the paper is organized as follows. Section \ref{sec:system_model} introduces the system model of Federated Learning, Over-the-Air Computation, and the retransmissions. Section \ref{sec:powercontrol} presents and solves the power control problem to minimize the \ac{MSE} between the desired and received sum. Section \ref{sec:convergence_analysis} provides worst-case analysis on the performance of AirReComp in terms of two upper bounds on the \ac{FL} loss function. These bounds are then used in Section \ref{sec:heuristic} to develop a heuristic for selecting the number of retransmissions that minimizes the \ac{FL} loss without exceeding communication and computation cost constraints. In Section \ref{sec:numerical_results}, the performances of AirReComp and the retransmission selection heuristic are numerically evaluated for non-convex loss functions. Finally, section \ref{sec:conclusion} concludes the work and discusses future work.

Notation: $z$ is a scalar, $\vtZ$ is a vector, and $\mathbf{Z}$ is a matrix. Element $i$ of vector $\vtZ$ is expressed as $z^{(i)}$. To denote element-wise operations of vectors, we overload the scalar equivalent, e.g. $\vtX/\vtY$ is the element-wise division of $\vtX$ and $\vtY$. $\Bar{z}$ denotes the complex conjugate. $\hat{z}$ denotes an estimate of $z$ and $\tilde{\vtZ}$ denotes a normalized version of $\vtZ$.


\section{System Model}
\label{sec:system_model}
In this section, we describe the system model and the AirReComp algorithm. For the reader's convenience, we have included a table with all key variables, see Table \ref{tab:notation}.

We consider a distributed \ac{ML} system consisting of $K$ single-antenna user devices each carrying a distinct dataset $\mathcal{D}_k$ and a single-antenna \ac{PS} which can be reached by all devices in a single hop, as illustrated in Fig. \ref{fig:model}. The objective of the system is to solve the following optimization problem
\begin{equation}
\label{eq:opt_problem}
    \vtW^* = \argmin_\vtW F(\vtW) = \argmin_\vtW \frac{1}{K} \sum_{k=1}^K F_k(\vtW),
\end{equation}
using the datasets at the user devices. The vector $\vtW \in \mathbb{R}^d$ is the $d\times 1$ parameter vector that defines the \ac{ML} model, $F(\vtW)$ is denoted the global loss function, and $F_k(\vtW)$ is a local loss function. The choice of $F_k(\vtW)$ describes the \ac{ML} task, where for instance binary cross entropy loss corresponds to a classification task and mean-squared error to regression. The results of this paper are applicable to any problem that can be expressed as \eqref{eq:opt_problem}, and as such $F_k(\vtW)$ can be seen as an arbitrary loss function

\begin{table}[t]
    \caption{Reference list of variables used in this paper. Ordered alphabetically and by case. }
    \label{tab:notation}
    \begin{center}
    \begin{tabular}{|c|c|}
        \hline
        \textbf{Variable} & \textbf{Interpretation} \\
        \hline
        $E$ & Number of local epochs in each communication round \\
        \hline
        $F(\vtW_n)$ & Federated Learning loss function \\
        \hline
        $K$ & Number of devices \\
        \hline
        $M$ & Number of uplink transmissions per communication round \\
        \hline
        $N$ & Total number of communication rounds \\
        \hline
        $\beta$ & Static learning rate \\
        \hline
        $d$ & Number of model parameters \\
        \hline
        $\eta$ & Post-transmission scalar \\
        \hline
        $h_k$ & Channel fading coefficient from device $k$ to PS \\
        \hline
        $L$ & Lipschitz smoothness parameter of $F(\vtW_n)$ \\
        \hline
        $\mu$ & Strong convexity parameter of $F(\vtW_n)$ \\
        \hline
        $\mu_{n,k}$ & Mean of local update at iteration $n$ and device $k$ \\
        \hline
        $p_k$ & Transmission power at $k$ \\
        \hline
        $\sigma_{n,k}^2$ & Variance of local update at iteration $n$ and device $k$ \\
        \hline
        $\sigma_z^2$ & Variance of additive white Gaussian noise \\
        \hline
        $\vtSigma$ & Variance bound on local and global model difference \\
        \hline
        $\vtW_n$ & Global model parameters at iteration $n$ \\
        \hline
        $\Delta\vtW_n$ & Global model update at iteration $n$ \\
        \hline
        $\vtW_{n,k}$ & Local model parameters at iteration $n$ and device $k$ \\
        \hline
        $\Delta\vtW_{n,k}$ & Local model update \\
        \hline
        $\Delta\tilde{\vtW}_{n,k}$ & Normalized local model update \\
        \hline
        $\Delta\hat{\vtW}_{n,k}$ & Estimate of local model update \\
        \hline
        $\vtX_k$ & Amplitude of transmitted signal from device $k$ \\
        \hline
        $\vtZ$ & Additive white Gaussian noise \\
        \hline
    \end{tabular}
    \end{center}
\end{table}

The uplink wireless channel is modeled as a block-fading \ac{MAC} with additive noise \cite{zhu2019broadband,sery2020analog}. If the $K$ users simultaneously transmit a vector $\vtX_k \in \mathbb{R}^d$ over the \ac{MAC}, the \ac{PS} receives
\begin{equation}
\label{eq:MAC}
    \vtY = \sum_{k=1}^Kh_k\vtX_k+\vtZ,
\end{equation}
where $h_k \in \mathbb{C}$ denotes the channel coefficient from device $k$ to the \ac{PS} and $\vtZ \in \mathbb{C}^d$ denotes \ac{AWGN} with variance $\sigma_z^2$. Such a model is a good approximation of reality if all users transmit simultaneously and synchronously over a wireless channel \cite{abari2016over, cao2020optimized, zhu2019broadband}. For simplicity, we assume error-free broadcast transmission in the downlink, which is an acceptable approximation for most practical scenarios since the \ac{PS} generally has much greater communication capability than the user devices \cite{tran2019federated}.

\subsection{Federated Learning Algorithm}

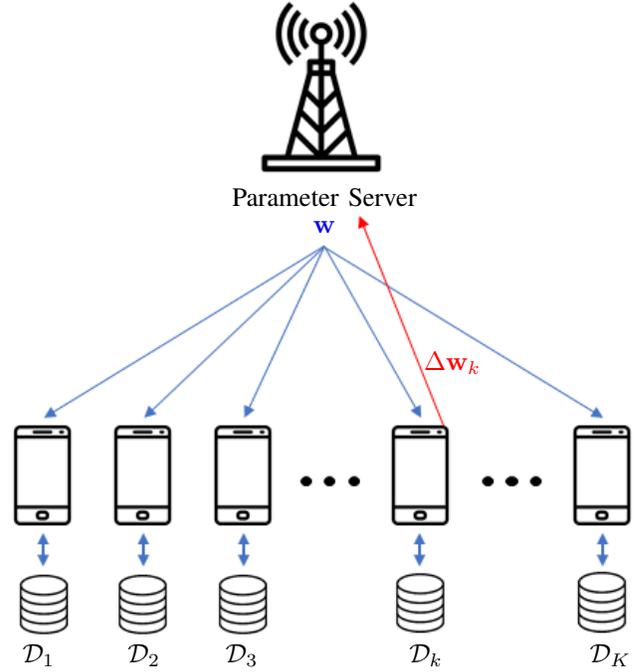
\begin{figure}[t]
    \begin{tikzpicture}
        \node[inner sep=0pt] (russell) at (0,0)
            {\includegraphics[width=8.5cm]{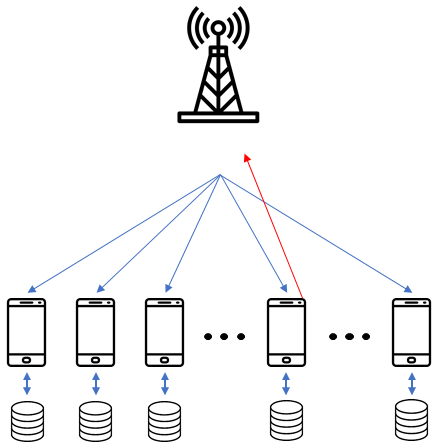}};
        \node[] at (0,1.6) {Parameter Server};
        \node[text=red] at (1.7,-0.6) {$\Delta\vtW_k$};
        \node[text=blue] at (0,1.2) {$\vtW$};
        \node[] at (-3.8,-4.5) {$\mathcal{D}_1$};
        \node[] at (-2.4,-4.5) {$\mathcal{D}_2$};
        \node[] at (-1.1,-4.5) {$\mathcal{D}_3$};
        \node[] at (1.35,-4.5) {$\mathcal{D}_k$};
        \node[] at (3.8,-4.5) {$\mathcal{D}_K$};
    \end{tikzpicture}
    \caption{Illustration of the Parameter Server and wireless network setup used throughout this paper. The training dataset is distributed among the User Devices, and all training takes place at the devices. The role of the Parameter Server is to orchestrate the communication of model updates from and to the User Devices.}
    \label{fig:model}
\end{figure}

\ac{FL} is an iterative algorithm to solve \eqref{eq:opt_problem}, where each iteration is denoted a communication round and consists of downlink broadcast, model training at the user devices, and uplink aggregation. Each of these communication rounds is associated with a computational cost $C_t$ (model training) and a communication cost $C_u$ (uplink aggregation). Considering negligible downlink communication cost, the total cost of \ac{FL} is $(C_t+C_u)N$, where $N$ is the total number of communication rounds.

Communication round $n$ starts when the \ac{PS} broadcasts the global model $\vtW_n$ to all user devices in the downlink \cite{chen2020joint}. Upon receiving the model, user device $k$ solves the local problem
\begin{equation}
\label{eq:opt_problem2}
    \vtW_{n,k}^* = \argmin_{\vtW} F_k(\vtW) =\argmin_{\vtW}\sum_{\vtU_i\in \mathcal{D}_k}l(\vtW, \vtU_{i})
\end{equation}
where $\vtU_{i}$ denotes one training sample and $l(\vtW, \vtU_{i})$ is the sample-wise loss function. Generally, \eqref{eq:opt_problem2} can not be solved exactly. Instead, each device runs $E$ epochs of gradient descent to approximately solve \eqref{eq:opt_problem2} as follows
\begin{equation}
\label{eq:local_training}
    \vtW_{n,k}(i) \leftarrow \vtW_{n,k}(i-1) - \beta\nabla F_k(\vtW_{n,k}(i-1)), \forall i =1,..,E,
\end{equation}
where the first iteration is based on the global model, i.e. $\vtW_{n,k}(0) = \vtW_n$, and $\beta$ is the step size. After executing $E$ epochs, device $k$ calculates a local model update as $\Delta \vtW_{n,k} = (\vtW_n - \vtW_{n,k}(E))/\beta$. After all local model updates have been computed, they are transmitted in the uplink to the \ac{PS}.

At the \ac{PS}, the local model updates are aggregated to form a global model update, written as
\begin{equation}
\label{eq:global_update}
    \Delta\vtW_n = \frac{1}{K}\sum_{k=1}^K\Delta \vtW_{n,k}.
\end{equation}
Finally, the \ac{PS} concludes the communication round by generating the next iteration of the model parameters
\begin{equation}
\label{eq:update_rule}
    \vtW_{n+1} = \vtW_n - \beta\Delta\vtW_n.
\end{equation}
The algorithm repeats until $\vtW_N$ is generated as the final model. After all $N$ communication rounds have completed, the \ac{PS} broadcasts $\vtW_N$ to the user devices, so that every device has the ability to perform inference with the model.

\subsection{Over-the-Air Computation Protocol}
\label{sec:aircomp}
In the uplink aggregation step of \ac{FL}, see \eqref{eq:global_update}, the \ac{PS} reconstructs the sum of $K$ model updates. In this section, we describe how this is achieved by \ac{AirComp}.

To start, the devices normalize the model updates according to
\begin{equation}
\label{eq:normalize}
    \Delta \tilde{\vtW}_{n,k} = \frac{\Delta \vtW_{n,k}-\mu_{n,k}}{\sigma_{n,k}},
\end{equation}
where $\mu_{n,k}$ and $\sigma_{n,k}$ are calculated as
\begin{equation}
\begin{split}
    \mu_{n,k} = \frac{1}{d}\sum_{i=1}^d\Delta w_{n,k}^{(i)} \\
    \sigma_{n,k}^2 = \frac{1}{d-1}\sum_{i=1}^d\left|\Delta w_{n,k}^{(i)} - \mu_{n,k}\right|^2,
\end{split}
\end{equation}
where $\Delta w_{n,k}^{(i)}$ is element $i$ of vector $\Delta \vtW_{n,k}$. For denormalization purposes, the values $\mu_{n,k}$ and $\sigma_{n,k}^2$ are transmitted to the \ac{PS} over a separate control channel. Note that the overhead of transmitting 2 scalars per device is negligible for most \ac{FL} problems. After normalizing, the model updates are communicated over the \ac{MAC} defined in \eqref{eq:MAC} as
\begin{equation}
\label{eq:tx_signal}
    \vtX_k = \Delta \tilde{\vtW}_{n,k}\frac{\Bar{h}_k}{|h_k|}\sqrt{p_k}.
\end{equation}
By combining \eqref{eq:MAC} and \eqref{eq:tx_signal} we get the received value at the \ac{PS}
\begin{equation}
\label{eq:received_signal}
    \vtY = \sum_{k=1}^K|h_k|\sqrt{p_k}\Delta \tilde{\vtW}_{n,k}+\vtZ.
\end{equation}
Ideally, the transmission powers would be chosen as $p_k = 1/|h_k|^2$, which would completely compensate for the fading effect. However, with a natural constraint on the maximum transmission power, $p_k = 1/|h_k|^2$ might be impossible to achieve. Because of this limitation and due to the additive noise, the \ac{PS} can never perfectly reconstruct $\Delta \tilde{\vtW}_n$. Instead, it estimates $\Delta \tilde{\vtW}_n$ by dividing the received signal by a post-transmission scalar $\sqrt{\eta}$ and the number of devices $K$
\begin{equation}
\label{eq:comac_estimate}
    \overline{\vtY} = \frac{\vtY}{\sqrt{\eta}K} = \sum_{k=1}^K \frac{|h_k|\Delta\tilde{\vtW}_{n,k}\sqrt{p_k}}{\sqrt{\eta}K} + \frac{\vtZ}{\sqrt{\eta}K}.
\end{equation}
Coupled with the transmission powers, $\sqrt{\eta}$ has an important role. We see that the ideal choice of the transmission powers is now $\sqrt{p_k}=\sqrt{\eta}/|h_k|$. As such, the selection of a small $\eta$ will reduce the amount of energy required to invert a channel and thereby reduce the fading error. However, lowering $\eta$ will also increase the relative power of the noise. Therefore, the post-transmission scalar $\sqrt{\eta}$ will play the role of a tradeoff parameter between the fading error and the noise-induced error \cite{cao2020optimized}.

In this work we propose AirReComp, which considers retransmissions in the uplink aggregation step. Specifically, the devices transmit the same values in the uplink $M$ times. This way, the signal part of \eqref{eq:comac_estimate} will combine constructively, while the additive noise is different for each transmission. After receiving $M$ values, the \ac{PS} forms its estimate by calculating their arithmetic mean
\begin{equation}
\label{eq:update_estimate}
    \overline{\vtY} = \frac{\sum_{m=1}^M\vtY_m}{M\sqrt{\eta}K} = \sum_{k=1}^K \frac{|h_k|\Delta \tilde{\vtW}_{n,k}\sqrt{p_k}}{\sqrt{\eta}K} + \sum_{m=1}^M\frac{\vtZ_m}{M\sqrt{\eta}K}.
\end{equation}
Here, we have assumed static fading coefficients $h_k$ for the duration of the $M$ transmissions. To finalize the global update, the \ac{PS} takes the real part of $\overline{\vtY}$ and denormalizes the result according to
\begin{equation}
\label{eq:global_ota_update}
    \Delta \hat{\vtW}_n = \operatorname{Re}(\overline{\vtY})\frac{\sum_{k=1}^K\sigma_{n,k}}{K} + \frac{\sum_{k=1}^K\mu_{n,k}}{K}.
\end{equation}
This final estimate is then used in the global update step to generate the next iteration of the model as
\begin{equation}
\label{eq:update_rule_ota}
    \vtW_{n+1} = \vtW_n - \beta\Delta\hat{\vtW}_n.
\end{equation}
The whole AirReComp process is summarized in Algorithm \ref{alg:airrecomp}.

\begin{remark}
Since $\Delta \tilde{\vtW}_{n,k}$ has unit variance, $p_k$ denotes the expected transmission power for each element of the model update, as follows
\begin{equation}
    \mathbb{E}[||\vtX_k||^2] = \mathbb{E}[\Delta \tilde{\vtW}_{n,k}^T\Delta \tilde{\vtW}_{n,k}]p_k = dp_k.
\end{equation}
\end{remark}
\begin{remark}
If the \ac{PS} only has access to the sum of the normalized model updates $\sum_{k=1}^K\Delta \tilde{\vtW}_{n,k}$, it is not possible to perfectly reconstruct $\Delta \vtW_n$ in general. However, the distortion caused by transmitting raw model updates with \ac{AirComp} is generally much greater than the distortion caused by the denormalization scheme we opt for here. In the specific case of $\sigma_1=\sigma_2=...=\sigma_K$, this scheme offers perfect reconstruction.
\end{remark}

\begin{algorithm}[t]
\caption{AirReComp}\label{alg:airrecomp}
\begin{algorithmic}
\ps
    \State initialize $\vtW_0$
\endps
\For{each round $n=0,1,..,N-1$}
    \ps
        \State broadcast $\vtW_n$ to devices
    \endps
    \device 
        \State $\vtW_{n,k}(E) \gets$ Equation \eqref{eq:local_training}
        \State $\Delta \vtW_{n,k} \gets (\vtW_n - \vtW_{n,k}(E))/\beta$
        \State $\Delta \tilde{\vtW}_{n,k} \gets$ Equation \eqref{eq:normalize}
        \State $\vtX_k \gets$ Equation \eqref{eq:tx_signal}
        \State transmit $\mu_{n,k}$ and $\sigma_{n,k}$ to server
    \enddevice
    \For{each $m=0,1,...,M-1$}
        \devices
            \State transmit $\vtX_k$ to server
        \enddevices
    \EndFor
    \ps
        \State $\overline{\vtY} \gets$ Equation \eqref{eq:update_estimate}
        \State $\Delta \hat{\vtW}_n \gets$ Equation \eqref{eq:global_ota_update}
        \State $\vtW_{n+1} \gets \vtW_n - \beta\Delta\hat{\vtW}_n$
    \endps
\EndFor
\end{algorithmic}
\end{algorithm}


\section{Power Control}
\label{sec:powercontrol}
In this section, we consider a power control problem to minimize the mean-squared estimation error defined as
\begin{equation}
\label{eq:mse}
    \mathbb{E}[(\Delta\vtW_n-\Delta \hat{\vtW}_n)^2],
\end{equation}
where the expectation is taken over the \ac{AWGN}, and $\Delta\vtW_n$ and $\Delta \hat{\vtW}_n$ are defined in \eqref{eq:global_update} and \eqref{eq:global_ota_update} respectively. For mathematical tractability, we assume that these gradient elements are \ac{IID} \cite{cao2020optimized, liu2020over, zang2020over}. Additionally, we ignore the effect of normalization by considering $\mu_{n,k}=0$ and $\sigma_{n,k}=1$. To perform the minimization, we seek the optimal choice of the transmission powers $\sqrt{p_k}$ and the post-transmission scalar $\sqrt{\eta}$. Since we consider static fading coefficients, the power control problem only has to be solved once per communication round. To model the limited transmission power of the devices, we consider a peak power constraint
\begin{equation}
    \mathbb{E}[|x_k^{(i)}|^2] = \mathbb{E}\left[|\Delta \tilde{w}_{n,k}^{(i)}\frac{\Bar{h}_k}{|h_k|}\sqrt{p_k}|^2\right] = p_k\leq P_{\text{max}}\ \forall k,
\end{equation}
where the expectation is taken over the model update $\Delta\tilde{w}_{n,k}^{(i)}$ and $P_{\text{max}}$ is the maximum transmission power. The minimization of \eqref{eq:mse} is formulated as
\begin{equation}
\label{eq:optproblem}
\begin{split}
    \begin{aligned}
    &\min_{\vtP, \eta} \\&\mathbb{E}\left[ \left(\sum_{k=1}^K \frac{|h_k|\Delta\vtW_{n,k}\sqrt{p_k}}{\sqrt{\eta}K} + \sum_{m=1}^M\frac{\operatorname{Re}(\vtZ_m)}{M\sqrt{\eta}K} - \sum_{k=1}^K\frac{\Delta\vtW_n}{K}\right)^2 \right]\\
    &\textrm{s.t.} \quad  p_k \leq P_{\text{max}},\ \forall k\in\{1<k<K\}.
    \end{aligned}
\end{split}
\end{equation}
Note that the number of transmissions $M$ is given as an input parameter and is selected before the power control problem is solved.
\begin{proposition}
\label{proposition:powercontrol}
Problem \eqref{eq:optproblem} has a unique solution. The optimal post-transmission scalar is given by the solution to the $K$ subproblems
\begin{equation}
\label{eq:eta_star}
    \eta^* = \min_k\Tilde{\eta}_k,
\end{equation}
where
\begin{equation}
\label{eq:etak}
    \Tilde{\eta}_k = \left(\frac{\sum_{j=1}^k|h_j|^2\overline{P}+\sigma_z^2/M}{\sum_{j=1}^k|h_j|\sqrt{\overline{P}}}\right)^2.
\end{equation}
The optimal transmission powers are
\begin{equation}
\label{eq:opt_pre}
    p_k^* = \operatorname{min}\left(\overline{P}, \frac{\eta^*}{|h_k|^2}\right).
\end{equation}
\end{proposition}
The proof of Proposition \ref{proposition:powercontrol} follows the proof in \cite{cao2020optimized} and is omitted from this paper.
\begin{remark}
\label{rem:eta}
From \eqref{eq:etak}, we see that the post-transmission scalar $\sqrt{\eta}$ assumes a lower value when more retransmissions are used. As we increase the number of retransmissions, the \ac{SNR} increases and consequently, the noise-induced error reduces. Therefore, the fading error becomes dominant and the optimal post-transmission scalar $\eta^*$ is lowered to improve it.
\end{remark}


\section{Convergence Analysis}
\label{sec:convergence_analysis}
In this section, we analyze the learning performance of Algorithm \ref{alg:airrecomp}. As before, we assume that the updates are zero-mean and unit variance, i.e. $\mu_{n,k}=0$, $\sigma_{n,k}=1$. Additionally, we assume that there is only one epoch of local training in each communication round ($E=1$). The performance is measured as the gap between the \ac{FL} loss gap at iteration $n$, defined as
\begin{equation}
\label{eq:excessloss}
    \mathbb{E}[F(\vtW_n)]-F(\vtW^*).
\end{equation}
We derive two upper bounds on this loss gap, one for strongly-convex functions and one for convex functions. For both bounds, we use the following well-known lemma \cite{sery2020analog, nesterov2003introductory}.
\begin{lemma}
\label{lemma:convex_smooth}
Let $F(\vtX) : \mathbb{R}^d\rightarrow\mathbb{R}$ be a convex function with L-Lipschitz gradient. Then, the following inequality holds:
\begin{equation}
\label{eq:convex_smooth}
    F(\vtY)-F(\vtX)-\nabla F(\vtX)^T(\vtY-\vtX) \leq \frac{L}{2}||\vtX-\vtY||^2.
\end{equation}
\end{lemma}
Additionally, we make an assumption on the similarity of the local model updates $\Delta\vtW_{n,k}$ and the global model update $\Delta\vtW_n$ \cite{cao2021optimized, zhang2021gradient}.
\begin{assumption}
\label{as:variance_bound}
The local model updates $\Delta\vtW_{n,k}$ are assumed to be independent and unbiased estimates of the global model update $\Delta\vtW_{n}$.
\begin{equation}
    \label{eq:unbiased}
    \mathbb{E}[\Delta\vtW_{n,k}] = \Delta\vtW_{n},\ \forall k\in\{1,2,\dots,K\}.
\end{equation}
The local model updates and the global gradient update are in general different. The difference has coordinate bounded variance:
\begin{equation}
    \label{eq:coordinate_bounded}
    \mathbb{E}[(\Delta w_{n,k}^{(i)}-\Delta w_{n}^{(i)})^2] \leq (\sigma^{(i)})^2,
\end{equation}
\end{assumption}
where $\Delta w_{n,k}^{(i)}$ is the $i$-th element of $\Delta \vtW_{n,k}$, and $(\sigma^{(i)})^2$ are the element-wise upper bounds. We will also use $\vtSigma\in\mathbb{R}^d$ to denote the vector of variance bounds.

\subsection{Strongly-convex loss}
In this subsection, we assume that the \ac{FL} loss is $\mu$-strongly convex. For such a loss, we use the following lemma \cite{sery2020analog, nesterov2003introductory}:
\begin{lemma}
\label{lemma:stronglyconvex_smooth}
Let $F(\vtX) : \mathbb{R}^d\rightarrow\mathbb{R}$ be a $\mu$-strongly convex function with $L$-Lipschitz gradient. Then, the following inequality holds:
\begin{equation}
\label{eq:strongly_smooth}
\begin{split}
    \left(\nabla F(\vtX)-\nabla F(\vtY)\right)^T\left(\vtX - \vtY\right) \geq \\\frac{\mu L}{\mu+L}||\vtX-\vtY||^2 + \frac{1}{\mu+L}||\nabla F(\vtX) - \nabla F(\vtY)||^2.
\end{split}
\end{equation}
\end{lemma}
Before we are ready to state the upper bound, we must also assert that the fixed step size $\beta$ has been selected to be sufficiently small for convergence. For AirReComp, the norm of the gradient is dependent on the power with which the signal is received, as such, the step size must be upper bounded as a function of the power control.
\begin{assumption}
\label{as:learning_rate}
Let the fixed step size $\beta$ be:
\begin{equation}
\label{eq:learning_rate}
    \beta < \operatorname{min}\left(\frac{K\sqrt{\eta}(\mu+L)}{2\mu L\sum_{k=1}^K\sqrt{p_k}|h_k|}, \frac{2\sqrt{\eta}}{\mu+L}\frac{\sum_{k=1}^K\sqrt{p_k}|h_k|}{\sum_{k=1}^Kp_k|h_k|^2}\right).
\end{equation}
\end{assumption}
We are now ready to give the first upper bound on the \ac{FL} loss function \eqref{eq:excessloss} given the update described in \eqref{eq:global_ota_update} and \eqref{eq:update_rule_ota}.
\begin{proposition}
\label{proposition:stronglyconvex}
Let
\begin{equation}
\label{eq:c1_constant}
    c_1 \coloneqq \frac{\sum_{k=1}^K\sqrt{p_k}|h_k|}{\sqrt{\eta}},
\end{equation}
\begin{equation}
\label{eq:convergence_constant}
    c_2 \coloneqq 1-\frac{2\beta}{K}\frac{\mu L}{\mu+L}c_1,
\end{equation}
and
\begin{equation}
\label{eq:c3_constant}
    c_3 \coloneqq ||\vtSigma||^2\frac{\sum_{k=1}^Kp_k|h_k|^2}{K\eta}+\frac{d\sigma_z^2}{MK^2\eta}.
\end{equation}
Then the \ac{FL} loss is upper bounded by
\begin{equation}
\label{eq:bound}
\begin{split}
    \mathbb{E}[F(\vtW_n)] - F(\vtW^*) \leq \frac{L}{2}c_2^n\mathbb{E}[r_0^2]+\frac{\beta^2L}{2(1-c_2)}c_3,
\end{split}
\end{equation}
where $r_0=||\vtW_0-\vtW^*||$ is the distance between the initial weight vector and the optimal one, $\vtSigma$ is a vector of the coordinate bounded variances from \eqref{eq:coordinate_bounded}, and $d$ is the number of model parameters.
\end{proposition}
\begin{proof}
The proof is provided in Appendix \ref{appendix:stronglyconvex}.
\end{proof}
We refer to the first term of the RHS of \eqref{eq:bound} as the \textit{diminishing term}, because it approaches zero if $n \rightarrow \infty$. Along the same line, we refer to the other terms as the \textit{post-convergence terms} because they remain even if $n \rightarrow \infty$. The implications of Proposition \ref{proposition:stronglyconvex} are given in Section \ref{sec:convergence_discussion}.

\subsection{Convex loss}
In this subsection, we relax the assumption on strong convexity and develop a bound for Lipschitz smooth and convex loss functions. For this bound, we need a different guarantee on the fixed step size than for the strongly convex case.
\begin{assumption}
\label{as:learning_rate2}
The fixed step size $\beta$ is selected to satisfy:
\begin{equation}
\label{eq:learning_rate2}
    0 < \beta < \frac{\sqrt{\eta}}{L}\frac{\sum_{k=1}^K\sqrt{p_k}|h_k|}{\sum_{k=1}^Kp_k|h_k|^2}.
\end{equation}
\end{assumption}
\begin{proposition}
\label{proposition:convex}
Consider Assumption \ref{as:variance_bound} and \ref{as:learning_rate2}. Then the \ac{FL} loss is upper bounded by
\begin{equation}
\label{eq:bound2}
\begin{split}
    \mathbb{E}[F(\vtW_n)] - F(\vtW^*) \leq \frac{K}{2n\beta c_1}\mathbb{E}[r_0^2]+\frac{\beta}{2}(\frac{K}{c_1} + L\beta)c_3,
\end{split}
\end{equation}
where $c_1$ and $c_3$ are defined in \eqref{eq:c1_constant} and \eqref{eq:c3_constant} respectively.
\end{proposition}
\begin{proof}
The proof is provided in Appendix \ref{appendix:convex}.
\end{proof}
Just as for Proposition \ref{proposition:stronglyconvex}, we refer to the first term of the RHS of \eqref{eq:bound} as the \textit{diminishing term}, and the other terms as the \textit{post-convergence terms}. The implications of Proposition \ref{proposition:convex} are given in Section \ref{sec:convergence_discussion}.

\subsection{Discussion on Proposition \ref{proposition:stronglyconvex} and \ref{proposition:convex}}
\label{sec:convergence_discussion}
In this section, we discuss Proposition \ref{proposition:stronglyconvex} and \ref{proposition:convex}. Since the propositions are upper bounds, we are discussing the worst-case properties of the \ac{FL} loss using AirReComp. We are specifically interested in the impact of the number of retransmissions.

\subsubsection{Convergence Rate}
As a shorthand, let $x_n$ equal the diminishing term of \eqref{eq:bound} or \eqref{eq:bound2} depending on which bound we are considering. Then, suppose the following holds
\begin{equation}
    \lim_{n\rightarrow\infty}\frac{|x_{n+1}-L|}{|x_n-L|} = \alpha.
\end{equation}
If $\alpha = 1$, we say that the worst-case convergence rate of AirReComp is sublinear, and if $0 < \alpha < 1$, we say that the worst-case convergence rate is linear.

For the strongly convex bound, we substitute $x_n$ with the right-hand-side of \eqref{eq:bound} and $L$ with the post-convergence terms. Basic algebra then tells us that $\alpha=c_2$. Since $c_2 < 1$, AirReComp has linear convergence rate for strongly-convex loss functions. Additionally, the rate of convergence is increasing with decreasing $c_2$. Since $c_2$ is decreasing in $c_1$, we can also say that the rate of convergence is increasing with $c_1$.

For the convex bound, the same procedure yields $\alpha=1$. Therefore, AirReComp has sublinear convergence rate for convex loss functions. Since $\alpha$ is constant with respect to $M$, this indicates that the asymptotic rate of convergence is independent of the number of retransmissions. However, if we consider a finite $n$, the diminishing term of \eqref{eq:bound2} is decreasing in $c_1$. Therefore, we say that the non-asymptotic convergence rate of AirReComp with convex loss functions is increasing in $c_1$. With these clarifications, we can state the following corollary.

\begin{corollary}
The worst-case convergence rate of AirReComp is increasing in $M$.
\end{corollary}
\begin{proof}
As clarified previously, the convergence rate is increasing in $c_1$ for both bounds. Therefore, we need to establish that $c_1$ is increasing in $M$. From \eqref{eq:c1_constant}, we see that $c_1$ depends on the power control policy via the post-transmission scalar $\eta$ and the transmission powers $p_k$. We already know that $\eta$ is decreasing in $M$ from Remark \ref{rem:eta}. For the transmission powers, we see in \eqref{eq:opt_pre} that the transmission powers take one of two values:
\begin{enumerate}
    \item The device inverts its channel with $p_k = \eta/|h_k|^2$. These devices contribute $1$ to the sum in $c_1$ regardless of $M$;
    \item The device transmits with maximum power $p_k = \overline{P}$. These devices contribute $\sqrt{\overline{P}}|h_k|/\sqrt{\eta}$ to the sum in $c_1$. Since $\eta$ is decreasing in $M$, these devices contribute more to $c_1$ for higher $M$.
\end{enumerate}
With these two statements we can readily conclude that $c_1$ is strictly increasing in $M$, given the condition that at least one device transmits with maximum power. In \cite{cao2020optimized}, they prove that this condition always holds with optimal power control. Thus, we can conclude that $c_1$ is strictly increasing in $M$.
\end{proof}

\subsubsection{Final error}
Since both bounds have post-convergence terms, the algorithm does not converge to a local optimum. Instead, the algorithm converges to a region around the optimum, where the expected remaining loss gap is given by the post-convergence terms. There are two reasons why AirReComp does not converge exactly. Firstly, the channel noise (characterized by $\sigma_z$) causes unavoidable errors which prevents exact convergence. Secondly, the difference between local and global model updates (characterized by $\vtSigma$) causes a global model update that differs from what is achieved in centralized gradient descent. This result aligns with what was found in \cite{cao2021optimized}.

Unfortunately, the post-convergence terms are not strictly decreasing in $M$\footnote{We made a mistake in \cite{hellstrom2021over} claiming the opposite.}. Depending on the parameters, the diminishing terms can both increase and decrease with $M$. As a rule of thumb, if $||\vtSigma||^2 >> \sigma_z^2$, the final error increases with $M$. However, due to resource constraints, training must generally be stopped long before convergence is reached. As such, the diminishing term tends to be dominant \cite{wang2018edge} and thus the convergence rate tends to be more relevant than the final error. It is worth to note that the gap could possibly be made to approach zero as $n\rightarrow\infty$ by considering learning rate scheduling, but in this paper we restrict ourselves to static $\beta$.


\section{Heuristic for selecting $M$}
\label{sec:heuristic}
As we have demonstrated in Section \ref{sec:convergence_analysis}, the inclusion of retransmissions in Over-the-Air \ac{FL} speeds up the worst-case convergence speed of \ac{FL} training. Simultaneously, the uplink communication cost is proportional to $M$. Therefore, $M$ acts as a tradeoff parameter between communication cost and \ac{FL} convergence speed. As such, the selection of $M$ is an important factor in model training on resource-constrained devices. 

To characterize this tradeoff more precisely, we consider a cost budget $\overline{C}$ for the whole training process from communication round $n=0$ to the final round $n=N$. For the purposes of this paper, the cost budget $\overline{C}$ corresponds to an abstract cost, but could be replaced with, e.g., energy or time depending on the application. As introduced in Section \ref{sec:system_model}, each communication round is associated with a computational cost $C_t$ and an uplink communication cost of $C_u$, which leads to the following constraint
\begin{equation}
\label{eq:cost_constraint}
    (C_t+MC_u)N \leq \overline{C}.
\end{equation}
Given that the system satisfies constraint \eqref{eq:cost_constraint}, we want to find the $M$ that minimizes the \ac{FL} loss after the final communication round $N$. We can formulate this as the following integer programming problem
\begin{equation}
\label{eq:optproblem2}
\begin{split}
    \begin{aligned}
    &\min_{M,N} \mathbb{E}[F(\vtW_N)]\\
    &\textrm{s.t.} \quad (C_t+MC_u)N \leq \overline{C}\\
    &M,N \in \mathbb{Z}^+.
    \end{aligned}
\end{split}
\end{equation}
Since $\mathbb{E}[F(\vtW_N)]$ is decreasing in N, the highest possible integer $N$ that does not break constraint \eqref{eq:cost_constraint} can be selected without loss of optimality. Thus we can transform \eqref{eq:optproblem2} to
\begin{equation}
\label{eq:optproblem3}
\begin{split}
    \begin{aligned}
    &\min_{M} \mathbb{E}[F(\vtW_N)]\\
    &\textrm{s.t.} \quad  N = \left\lfloor\frac{\overline{C}}{C_t+MC_u}\right\rfloor\\
    &M \in \mathbb{Z}^+,
    \end{aligned}
\end{split}
\end{equation}
which is easier than \eqref{eq:optproblem2}. However, in order to solve this problem, we require an exact expression of $\mathbb{E}[F(\vtW_N)]$. In the current \ac{ML} literature, no such expression is known. Instead, we propose to replace the objective function of \eqref{eq:optproblem3} with a proxy function that resembles $\mathbb{E}[F(\vtW_N)]$.

\subsection{Full bound proxy}
Two possible proxy functions for $\mathbb{E}[F(\vtW_N)]$ are the upper bounds on $\mathbb{E}[F(\vtW_N)] - F(\vtW^*)$ in \eqref{eq:bound} (Proposition \ref{proposition:stronglyconvex}) and \eqref{eq:bound2} (Proposition \ref{proposition:convex}). There are two things to note about this choice. 
\begin{itemize} 
    \item Initially, it might seem wrong to use $\mathbb{E}[F(\vtW_N)] - F(\vtW^*)$ as a proxy function for $\mathbb{E}[F(\vtW_N)]$, because the desired objective function is the expected loss, while our bounds are on the difference between the expected loss and the optimal loss. However, the optimal loss $F(\vtW^*)$ is constant with respect to $M$, so including it in the objective function has no impact on the optimal decision variable $M^*$.
    \item The upper bounds represent a worst-case analysis of the actual loss and we have no proof of tightness for these bounds. Therefore, this proxy function can at best be considered a heuristic for finding the optimal $M$ and the efficacy of the choice should be evaluated numerically, which we do in Section \ref{sec:numerical_results}.
\end{itemize}

\subsection{Diminishing term proxy}
The full bound proxy described above is the best approximation of $F(\vtW_N)$ we have available. However, to use it in optimization problem \eqref{eq:optproblem3}, rich knowledge about the dataset and loss function is required. Specifically, the Lipschitz smoothness constant $L$ and the strong convexity parameter $\mu$ must be known to calculate the full bound. The calculation of these constants requires access to the global dataset, whereas the devices are only carrying smaller local datasets. Additionally, for more complicated models such as \acp{DNN}, the calculation of $L$ is difficult and one can at best hope to estimate it \cite{fazlyab2019efficient}. As such, it is desirable to have a proxy function that can be calculated without the knowledge of $\mu$ and $L$.

In the strongly convex bound from Proposition \ref{proposition:stronglyconvex} \eqref{eq:bound}, the $\mu$ and $L$ constants are ubiquitous, which makes the bound difficult to calculate in practice. However, in \eqref{eq:bound2} from Proposition \ref{proposition:convex} there is of course no $\mu$ and the Lipschitz smoothness constant $L$ only appears in the post-convergence term. As such, the diminishing term of \eqref{eq:bound2} can be calculated using quite limited information. Additionally, we expect the diminishing term to be dominant since training over resource-constrained networks is generally stopped long before the algorithm has fully converged \cite{wang2018edge}. With this choice of proxy function, the optimization problem for finding $M$ becomes
\begin{equation}
\label{eq:optproblem4}
\begin{split}
    \begin{aligned}
    &\min_{M} \frac{K\sqrt{\eta}}{2N\beta \sum_{k=1}^K\sqrt{p_k}|h_k|}\mathbb{E}[r_0^2]\\
    &\textrm{s.t.} \quad  N = \left\lfloor\frac{\overline{C}}{C_t+MC_u}\right\rfloor\\
    &M \in \mathbb{Z}^+.
    \end{aligned}
\end{split}
\end{equation}
Here, every variable except $\mathbb{E}[r_0^2]$ is reasonable to know in practice. However, since $\mathbb{E}[r_0^2]$ is constant with respect to $M$ it can be dropped without changing the optimal decision variable $M^*$. We then have the problem
\begin{equation}
\label{eq:optproblem5}
\begin{split}
    \begin{aligned}
    &\min_{M} \frac{K\sqrt{\eta}}{2N\beta \sum_{k=1}^K\sqrt{p_k}|h_k|}\\
    &\textrm{s.t.} \quad  N = \left\lfloor\frac{\overline{C}}{C_t+MC_u}\right\rfloor\\
    &M \in \mathbb{Z}^+.
    \end{aligned}
\end{split}
\end{equation}
This final optimization problem is practical and computationally cheap to solve, however it is based on rough approximations of the desired problem \eqref{eq:optproblem2}. Therefore, the efficacy of the heuristic is not clear without experimental results, which we provide in Section \ref{sec:numerical_results}.


\section{Numerical Results}
\label{sec:numerical_results}
The performance of the proposed AirReComp system is now evaluated in terms of the model update estimation error, federated learning loss, and classification accuracy. Specifically, we have three goals with this section:
\begin{itemize}
    \item To demonstrate the need for retransmission-aware power control, by comparing our proposed solution with the state-of-the-art single transmission schemes proposed in \cite{cao2020optimized, liu2020over};
    \item To demonstrate that introducing retransmissions is also beneficial for non-convex loss functions. Note that our analytical results assumed convex loss functions.
    \item To demonstrate the viability of selecting $M$ using our proposed heuristic.
\end{itemize}

\begin{figure*}[t]
    \begin{tikzpicture}
        \node[inner sep=0pt] (russell) at (0,0)
            {\includegraphics[width=\textwidth]{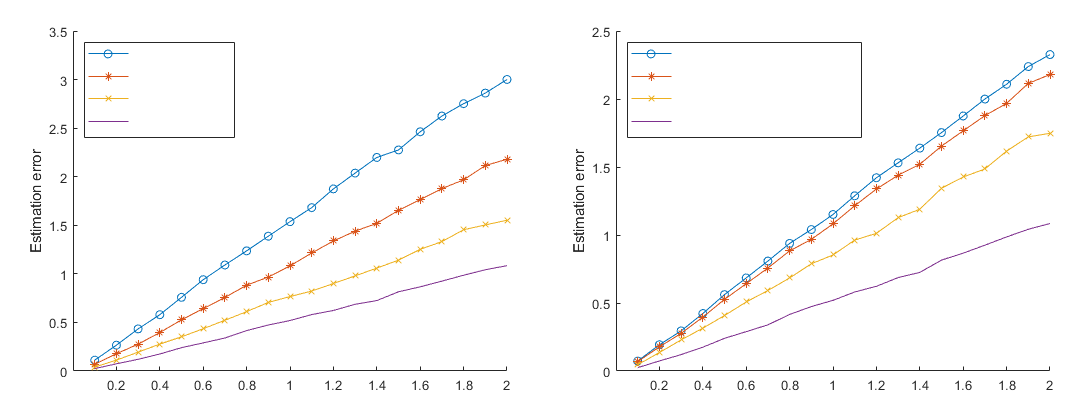}};
        \node[] at (-4.1,-3.3) {$\sigma_z$};
        \node[] at (-6.3,2.70) {$M=1$};
        \node[] at (-6.3,2.30) {$M=2$};
        \node[] at (-6.3,1.90) {$M=4$};
        \node[] at (-6.3,1.50) {$M=8$};
        \node[] at (5.3,-3.3) {$\sigma_z$};
        \node[] at (4.00,2.70) {Rtx-unaware $M=2$};
        \node[] at (3.82,2.30) {Rtx-aware $M=2$};
        \node[] at (4.00,1.90) {Rtx-unaware $M=8$};
        \node[] at (3.82,1.50) {Rtx-aware $M=8$};
    \end{tikzpicture}
    \caption{Estimation error evaluation of AirReComp. We consider $K=20$ devices and evaluate the squared estimation error. Left: The estimation error of a single transmission ($M=1$) is compared to using retransmissions ($M>1$). Note that even though \ac{SNR} scales linearly with the number of transmissions, the estimation error is not reduced as drastically. Right: The estimation error of optimal retransmission-aware power control is compared to a retransmission-unaware baseline. The results demonstrate the importance of designing the power control scheme with retransmissions in mind.}
    \label{fig:power_control}
\end{figure*}
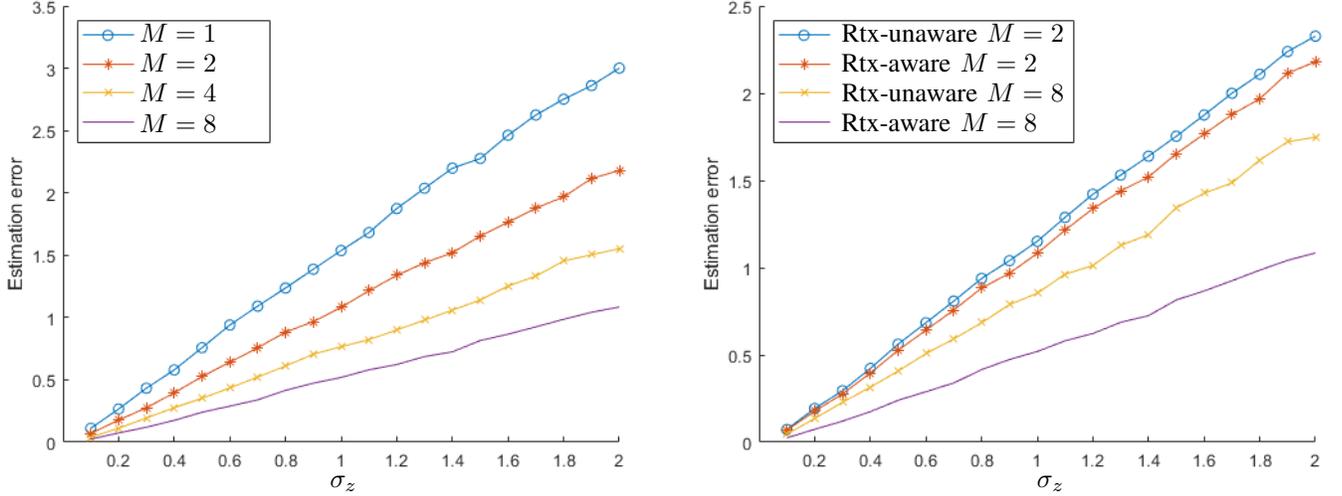
\begin{figure}[t]
    \begin{tikzpicture}
        \node[inner sep=0pt] (russell) at (0,0)
            {\includegraphics[width=8.5cm]{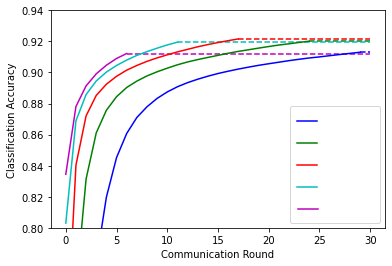}};
        \node[] at (3.3,0.3) {$M=1$};
        \node[] at (3.3,-0.18) {$M=2$};
        \node[] at (3.3,-0.66) {$M=4$};
        \node[] at (3.3,-1.14) {$M=8$};
        \node[] at (3.35,-1.62) {$M=16$};
    \end{tikzpicture}
    \caption{Classification accuracy of MNIST hand-written digit recognition problem. We consider $K=10$ devices and train fully-connected \acp{DNN} over a multiple access channel with fading. We compare five systems with different numbers of uplink retransmissions, where $M=1$ corresponds to using zero retransmissions. The results demonstrate an increased rate of convergence as the number of retransmissions increase.}
    \label{fig:accuracy}
\end{figure}
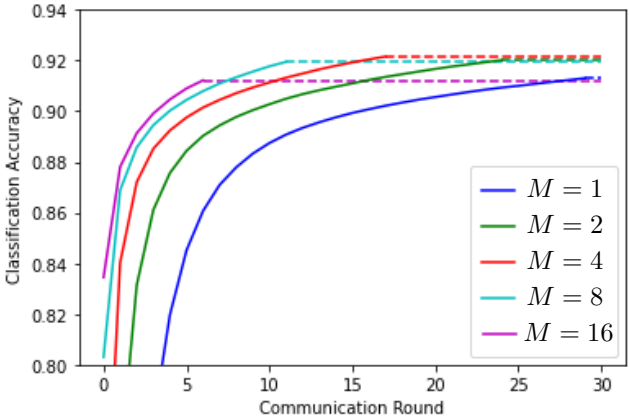
\begin{figure}[t]
    \begin{tikzpicture}
        \node[inner sep=0pt] (russell) at (0,0)
            {\includegraphics[width=8.5cm]{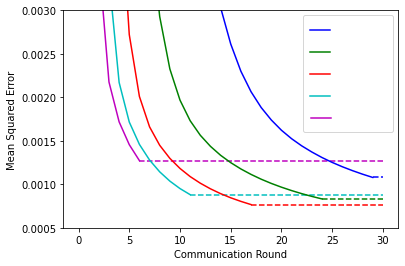}};
        \node[] at (3.3,2.25) {$M=1$};
        \node[] at (3.3,1.75) {$M=2$};
        \node[] at (3.3,1.28) {$M=4$};
        \node[] at (3.3,0.80) {$M=8$};
        \node[] at (3.39,0.35) {$M=16$};
    \end{tikzpicture}
    \caption{Mean-squared error loss of "Digital Demo Stockholm" water monitoring problem. We train fully-connected \acp{DNN} with MSE loss to predict the conductivity level of the water based on other measurements (temperature, salinity, etc...). The plot shows the normalized MSE between the test dataset conductivity and the predicted conductivity as a function of the communication rounds. This result demonstrates that AirReComp does well for non-convex regression in addition to classification.}
    \label{fig:accuracy2}
\end{figure}
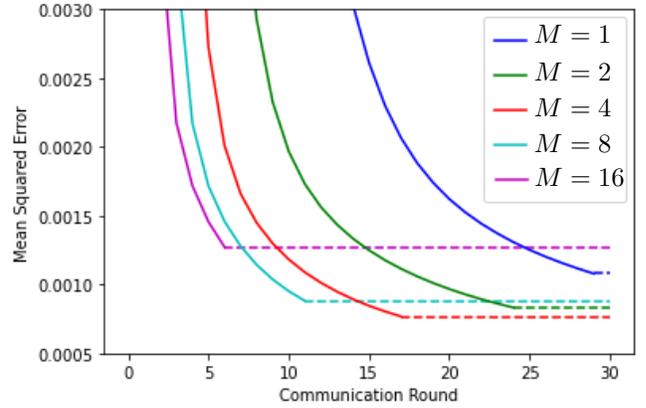

\subsection{Power Control}
\label{subsec:power_control}
In this subsection, we wish to evaluate the impact of our proposed power control scheme on the estimation error of the global update $\Delta\tilde{w}_k$. Specifically, we compare the \ac{MSE} of the $\Delta\tilde{w}_k$ for different choices of $M$, and compare the AirReComp power control scheme to the baseline solutions of \cite{cao2020optimized, liu2020over}. For this, we consider the transmission of randomly generated scalars instead of running a complete \ac{FL} simulation setup. For this simulation, we consider $K=20$ users and varying noise powers $\sigma_z^2$. To simulate the network, we generate channel coefficients according to unit Rayleigh fading $h_k \sim \mathcal{N}(0,1/2)+j\mathcal{N}(0,1/2)$ and additive noise components as $z \sim \mathcal{N}(0, \sigma_z^2)$. The transmitted scalars $\Delta\tilde{w}_k$ are generated according to the unit normal distribution, which matches the assumption in Section \ref{sec:powercontrol}. The \ac{PS} estimate of the arithmetic mean $\Delta \hat{\tilde{w}}$ is generated according to \eqref{eq:comac_estimate}, where the transmission powers $p_k$ and the post-transmission scalar $\sqrt{\eta}$ are selected according to Proposition \ref{proposition:powercontrol}. Upon calculating the estimate, it is compared to the true arithmetic mean of $\Delta\tilde{w}_k$ according to
\begin{equation}
    \text{MSE} = (\frac{1}{K}\sum_{k=1}^K\Delta\tilde{w}_k - \Delta \hat{\tilde{w}})^2.
\end{equation}
This process is repeated 20,000 times for each value of $\sigma_z^2$. The resulting \acp{MSE} are averaged to form the plot in Fig. \ref{fig:power_control}. In the left plot of the figure, the estimation error using different number of transmissions $M$ are illustrated. The plot demonstrates that the mean squared estimation error is approximately linear with the variance of the additive noise, regardless of $M$. In the right plot of Fig. \ref{fig:power_control}, we compare the AirReComp power control scheme to that proposed in \cite{cao2020optimized, liu2020over}. The baseline is optimal if there are no retransmissions, but the numerical results demonstrate that our scheme has a significantly lower estimation error when $M > 1$. The gap between AirReComp and the baseline is also increasing with $M$, which demonstrates the importance of designing the power control scheme with retransmissions in mind. From the left plot, it is clear that the reduction in estimation error is worse than proportional to $M$. Instead, the system using $M=8$ achieves approximately three times lower estimation error that the baseline of $M=1$. Compared to using a forward error-correcting code, this result is significantly worse. However, since such codes are not compatible with analog communication, retransmissions are a good first step towards enabling a communication-estimation trade-off.

\subsection{Federated Learning Convergence}
\label{sec:numerical_fl}
\begin{figure}[t]
    \begin{tikzpicture}
        \node[inner sep=0pt] (russell) at (0,0)
            {\includegraphics[width=8.5cm]{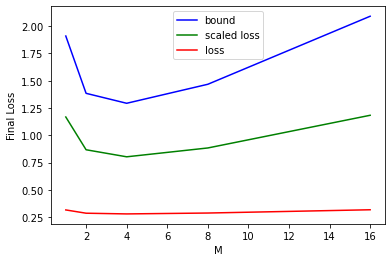}};
    \end{tikzpicture}
    \caption{Final loss after consuming the cost budget for the same systems as in Fig. \ref{fig:accuracy}. The empirical loss is compared to the diminishing term of our upper bound in \eqref{eq:optproblem3}. The heuristic matches the bound quite closely and selects the optimal $M$. }
    \label{fig:bound1}
\end{figure}
\begin{figure}[t]
    \begin{tikzpicture}
        \node[inner sep=0pt] (russell) at (0,0)
            {\includegraphics[width=8.5cm]{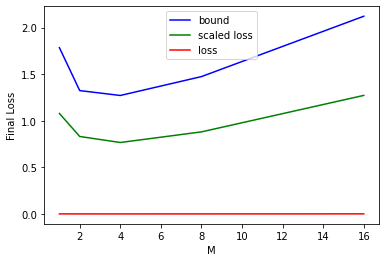}};
    \end{tikzpicture}
    \caption{Final loss after consuming the cost budget for a regression system for predicting the conductivity of freshwater. The heuristic selects $M=4$ which is not the optimal $M^*=8$. However, the difference in empirical performance between $M=4$ and $M=8$ is negligible, so $M=4$ is not a poor choice.}
    \label{fig:bound2}
\end{figure}
\begin{figure}[t]
    \begin{tikzpicture}
        \node[inner sep=0pt] (russell) at (0,0)
            {\includegraphics[width=8.5cm]{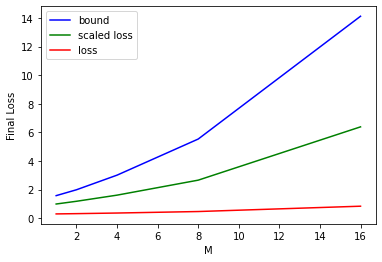}};
    \end{tikzpicture}
    \caption{Final loss after consuming the cost budget for a system where $M=1$ is the optimal choice. The heuristic is able to identify $M^*=1$ and can therefore avoid hurting the accuracy by introducing unwanted retransmissions.}
    \label{fig:bound3}
\end{figure}
\begin{figure}[t]
    \begin{tikzpicture}
        \node[inner sep=0pt] (russell) at (0,0)
            {\includegraphics[width=8cm]{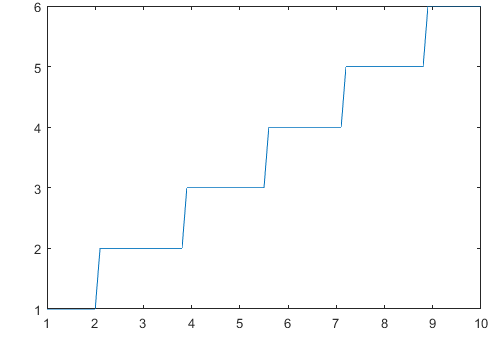}};
        \node[] at (-3.8,0.2) {$M^*$};
        \node[] at (0.5,-2.7) {$\sigma_z$};
    \end{tikzpicture}
    \caption{Optimal $M$ as predicted by our heuristic for different $\sigma_z^2$. We see that the optimal number of retransmissions are increasing with the additive noise power.}
    \label{fig:heuristic}
\end{figure}
In this subsection, we have two goals: To verify that the convergence rate is increasing in $M$ for non-convex loss functions and to show that $M>1$ can achieve improved performance over $M=1$ without exceeding the cost budget. For the \ac{FL} simulation, the network setup is identical to Subsection \ref{subsec:power_control}, except that the additive noise variance is fixed to $\sigma_z^2 = \sqrt{2K}$. The \ac{ML} task is multi-label classification on the MNIST dataset \cite{lecun1998gradient} with $|\mathcal{D}_k|=6000$ training samples per user device. The classifier is a \ac{DNN} which consists of an input layer of 784 nodes, a hidden layer with 100 nodes, and an output layer of 10 nodes. The network is trained with a static learning rate of $\beta=0.05$ with no dropout. We run 2 epochs ($E=2$) per communication round, with a budget of $\overline{C}=150$, training cost $C_t=4$ and communication cost $C_u=1$. The whole training process is repeated 50 times for each value of $M$ considered, these results are then averaged to get the plot in Fig. \ref{fig:accuracy}.

The results clearly demonstrate that the convergence rate is improved as additional retransmissions are introduced, even for non-convex loss functions. Notice that all systems stop training when the cost budget $\overline{C}$ is consumed, e.g. the system with $M=1$ stops after 50 communication rounds and the system with $M=2$ stops after 25 communication rounds. As such, we can see that AirReComp can outperform the baseline of using no retransmissions ($M=1$) in terms of classification accuracy, without incurring additional costs. Specifically, the system with $M=1$ achieved an average classification accuracy of $91.3\%$, while the best system with $M=4$ achieved an average classification accuracy of $92.1\%$.

In addition to training a classifier on the MNIST dataset, we also train a regression model to predict the conductivity of freshwater based on its temperature, percentage of dissolved oxygen, pH value, oxidation-reduction potential, and salinity. We have 30,000 datapoints in the dataset, of which we use 25,000 randomly selected data points for training and the remaining 5,000 as test data. The dataset used is the iWater dataset from "Digital Demo Stockholm" which is collected by sensors installed in lake Mälaren, Sweden \cite{iwater, iwater2}. For this simulation, we use exactly the same setup as the previous, except that the \ac{DNN} is reshaped to fit the new dataset. There is now an input layer of 4 nodes and an output layer of 1 node. Once again, we repeat the whole training process 50 times for each value of $M$ and the results are averaged to generate the plot in Fig. \ref{fig:accuracy2}. 
Similarly to the MNIST simulation, we see that the convergence rate is improved with increasing $M$, but now for a regression problem. Once again, the system stops after the cost budget $\overline{C}$ is consumed. The optimal system (lowest \ac{MSE}) with $M=4$ has a normalized \ac{MSE} of $0.0008$, whereas the baseline of $M=1$ achieved a normalized \ac{MSE} of $0.001$, so $M=4$ achieves approximately 20\% lower \ac{MSE} than the baseline.

\subsection{Selection of $M$}
In Section \ref{sec:heuristic}, we developed a heuristic for finding the $M$ which would yield the lowest loss without exceeding a cost budget $\overline{C}$. To verify the efficacy of using this heuristic, we have compared the final losses of the simulations in Section \ref{sec:numerical_fl} to the objective function of \eqref{eq:optproblem3}, we refer to this objective function as the \textit{bound heuristic}. In all plots of this section, we scale and translate the loss functions to give a better visual experience. Since we are only interested in the identification of the minimum value, this does not affect the result.

First for the MNIST dataset, we see in Fig. \ref{fig:bound1} that the shape of the bound heuristic resembles that of the empirical result. In this scenario, the heuristic would have selected $M=4$, which also ended up being the optimal choice for training. To generate this plot, we multiplied the empirical loss by 10 and then subtracted a constant of 2. The bound heuristic line was left without scaling or translation.

For the iWater dataset, we have a similar comparison in Fig. \ref{fig:bound2}. Both the bound and empirical loss function has a similar shape to Fig. \ref{fig:bound1}, despite being for a completely different \ac{ML} problem. However, the empirical loss function is of very different scale. In this figure, we multiplied the empirical loss by 1000, and did not perform any translation. The reason for this difference is because the regression problem is much easier than the classification problem and we therefore end up with lower losses in general. Since the shapes of the plots are the same, the bound heuristic correctly selects $M=4$, which is the optimal number of retransmissions. This result illustrates that the bound heuristic works well for both classification and regression problems.

There are cases where one should not employ retransmissions, because the cost of doing so outweighs the benefit of reduced estimation errors. To illustrate this, we ran an MNIST simulation with the same setup as the previous MNIST simulation but reduced the cost of computation to $C_t=0$, so that the the communication cost is relatively higher. Since the computation cost is lower, we also reduced the total budget to $\overline{C}=30$ so that the $M=1$ system gets to train for the same number of rounds (30 communication rounds). The result is illustrated in Fig. \ref{fig:bound3}. Here, we can see that the heuristic is able to correctly identify that retransmissions should not be used, since it would harm the final classification accuracy. Like the previous MNIST simulation, the empirical loss function is multiplied by 10 and then subtracted by 2 to generate the scaled loss.

Finally, we consider the relationship between the optimal number of transmissions $M^*$ and the variance of the additive noise $\sigma_z^2$. To illustrate this, we calculate $M^*$ using \eqref{eq:optproblem5} for all $\sigma_z = \{0,0.1,0.2,...,9.9,10\}$ to generate the plot in Fig. \ref{fig:heuristic}. In this simulation, we used the same parameters as in the simulation of Fig. \ref{fig:accuracy} and as such we see $M^*=4$ for $\sigma_z=\sqrt{2K}\approx6.3$. This result demonstrates that the optimal number of retransmissions increases with the power of the noise, which also seems intuitively correct.


\section{Conclusion}
\label{sec:conclusion}
In this paper, we propose retransmissions for Over-the-Air \ac{FL}, in a system we call AirReComp. Arguably, this is the first work to enable a trade-off between communication resources and convergence speed for Over-the-Air \ac{FL}. To improve the estimation error of AirReComp, we find a closed-form solution for optimal power control in the uplink. This power control solution shows that the number of retransmissions must be known by the transmitters to realize the MMSE estimator. We also prove two upper bounds on the \ac{FL} loss for the AirReComp system, both for strongly-convex and convex loss functions. These bounds show that the convergence rate of \ac{FL} is strictly increasing in the number of retransmissions. Finally, we provide a heuristic for selecting the optimal number of retransmissions based on the upper bounds. All of these results are verified numerically for non-convex loss functions by training \acp{DNN} with AirReComp, both in classification and regression tasks. The simulations also demonstrate that AirReComp can significantly outperform single uplink transmissions. There are conditions for which retransmissions are not desirable, we demonstrate that the heuristic can detect such conditions before training starts, which avoids introducing harmful retransmissions.

There are still interesting open problems on the reduction of estimation errors for Over-the-Air \ac{FL}, including:
\begin{itemize}
    \item \textbf{The consideration of fast-fading channels for AirReComp.} In this work, we consider a network with static channels, which restricts the applicability of the results to application areas with static sensors. In a fast-fading scenario, the problem changes substantially, especially the power control problem described in Section \ref{sec:powercontrol}.
    \item \textbf{Other methods of controlling the estimation error for Over-the-Air \ac{FL}.} For instance, one could consider the possibility of a distributed channel code. With analog communication, this appears to be inapplicable, but with the ideas of one-bit digital Over-the-Air Computation, there might be possibilities to explore in this direction \cite{zhu2020one}.
    \item \textbf{Tradeoff between transmission power and retransmissions.} Instead of focusing on adding retransmissions to improve the estimation error, one could consider changing the transmission power. Similar to the tradeoff explored in this work, there is a tradeoff between transmission power and convergence rate. Especially for low-powered IoT-devices, it would be interesting to analyze how much the transmission power could be reduced without significantly harming the \ac{FL} performance.
\end{itemize}


\appendices
\section{Proof of Proposition \ref{proposition:stronglyconvex}}
In Section \ref{sec:convergence_analysis} we assumed that the model updates $\Delta\vtW_n$ are already of zero mean and with unit variance. Therefore, no normalization is required. This can equivalently be expressed by saying that $\mu_{n,k}=0$ and $\sigma_{n,k}=1$. As a consequence of this assumption, we have
\begin{equation}
\label{eq:global_update_proof}
    \Delta \hat{\vtW}_n = \operatorname{Re}(\overline{\vtY}) = \sum_{k=1}^K \frac{|h_k|\Delta \vtW_{n,k}\sqrt{p_k}}{\sqrt{\eta}K} + \sum_{m=1}^M\frac{\operatorname{Re}(\vtZ_m)}{M\sqrt{\eta}K}.
\end{equation}

\label{appendix:stronglyconvex}
We start the proof by expressing the distance between the optimal global model $\vtW^*$ and the current global model $\vtW_n$ at communication round $n$ as
\begin{equation}
\label{eq:model_distance}
    r_n^2 \coloneqq ||\vtW_n-\vtW^*||^2.
\end{equation}
This distance can be related to the \ac{FL} loss function via Lemma \ref{lemma:convex_smooth} and Lemma \ref{lemma:stronglyconvex_smooth}. The plan for the proof is to utilize this relationship to form the upper bound. But before we get to that stage, we need to introduce the impact of AirReComp on the model update. To do so, we use \eqref{eq:update_rule_ota} with \eqref{eq:model_distance} to express
\begin{equation}
\label{eq:model_distance2}
\begin{split}
    r_{n+1}^2 \coloneqq ||\vtW_n-\vtW^*-\beta\Delta\hat{\vtW}_n||^2& \\= r_n^2 - 2\beta\Delta\hat{\vtW}_n^T(\vtW_n-\vtW^*) + \beta^2||\Delta\hat{\vtW}_n||^2&,
\end{split}
\end{equation}
where $\Delta\hat{\vtW}_n$ is the model update from \eqref{eq:global_update_proof}. Next, we take the expectation of \eqref{eq:model_distance2} with respect to $\Delta\vtW_{n,k}$ and $\vtZ_m$. To do that, we first need to determine $\mathbb{E}\left[\Delta\hat{\vtW}_n\right]$ and $\mathbb{E}\left[||\Delta\hat{\vtW}_n||^2\right]$. Beginning with $\mathbb{E}\left[\Delta\hat{\vtW}_n\right]$, we use \eqref{eq:global_update_proof} to get
\begin{equation}
\label{eq:expected_value_1}
    \mathbb{E}\left[\Delta\hat{\vtW}_n\right] = \mathbb{E}\left[\sum_{k=1}^K \frac{|h_k|\Delta \vtW_{n,k}\sqrt{p_k}}{\sqrt{\eta}K} + \sum_{m=1}^M\frac{\operatorname{Re}(\vtZ_m)}{M\sqrt{\eta}K}\right].
\end{equation}
Because of Assumption \ref{as:variance_bound} and since the noise $\vtZ_m$ is zero-mean, this simplifies to:
\begin{equation}
    \mathbb{E}\left[\Delta\hat{\vtW}_n\right] = \mathbb{E}[\Delta\vtW_n]\sum_{k=1}^K \frac{|h_k|\sqrt{p_k}}{\sqrt{\eta}K}.
\end{equation}
Since we assume there is only one epoch ($E=1$), the model update is the gradient of the global loss function, i.e., $\mathbb{E}[\Delta\vtW_n]=\mathbb{E}[\nabla F(\vtW_n)]$, which gives us
\begin{equation}
\label{eq:expected_value}
    \mathbb{E}\left[\Delta\hat{\vtW}_n\right] = \mathbb{E}[\nabla F(\vtW_n)]\sum_{k=1}^K \frac{|h_k|\sqrt{p_k}}{\sqrt{\eta}K}.
\end{equation}
Next, we find $\mathbb{E}\left[||\Delta\hat{\vtW}_n||^2\right]$, once again using \eqref{eq:global_update_proof}
\begin{equation}
\begin{split}
    &\mathbb{E}\left[||\Delta\hat{\vtW}_n||^2\right] = \\&\frac{1}{K^2\eta}\mathbb{E}\left[||\sum_{k=1}^K |h_k|\Delta \vtW_{n,k}\sqrt{p_k} + \sum_{m=1}^M\frac{\operatorname{Re}(\vtZ_m)}{M}||^2\right].
\end{split}
\end{equation}
The cross-multiplication term disappears because the noise is zero mean and independent from the model update
\begin{equation}
\label{eq:w_squared}
\begin{split}
    \mathbb{E}\left[||\Delta\hat{\vtW}_n||^2\right] = \frac{1}{K^2\eta}\mathbb{E}\left[||\sum_{k=1}^K \sqrt{p_k}|h_k|\Delta \vtW_{n,k}||^2\right] \\+\frac{1}{M^2K^2\eta} \mathbb{E}\left[||\sum_{m=1}^M\operatorname{Re}(\vtZ_m)||^2\right].
\end{split}
\end{equation}
The first term of \eqref{eq:w_squared} can be upper-bounded by the Cauchy-Schwartz inequality as follows
\begin{equation}
\label{eq:cauchy_schwartz}
    ||\sum_{k=1}^K \sqrt{p_k}|h_k|\Delta \vtW_{n,k}||^2 \leq K\sum_{i=1}^d\left(\sum_{k=1}^Kp_k|h_k|^2(\Delta w_{n,k}^{(i)})^2\right).
\end{equation}
Then, we apply our assumption on the local model updates from \eqref{eq:coordinate_bounded} to \eqref{eq:cauchy_schwartz} to get
\begin{equation}
\label{eq:cauchy_schwartz2}
    ||\sum_{k=1}^K \sqrt{p_k}|h_k|\Delta \vtW_{n,k}||^2 \leq K\sum_{k=1}^Kp_k|h_k|^2(||\Delta \vtW_n||^2 + ||\vtSigma||^2).
\end{equation}
Combining \eqref{eq:w_squared} and \eqref{eq:cauchy_schwartz2} gives us
\begin{equation}
\begin{split}
    &\mathbb{E}\left[||\Delta\hat{\vtW}_n||^2\right] \leq \\ &\frac{1}{K\eta}\mathbb{E}\left[\sum_{k=1}^K p_k|h_k|^2(||\Delta \vtW_n||^2 + ||\vtSigma||^2)\right] +\frac{d\sigma_z^2}{MK^2\eta},
\end{split}
\end{equation}
and since we assume $E=1$, we get
\begin{equation}
\label{eq:expected_square}
\begin{split}
    \mathbb{E}\left[||\Delta\hat{\vtW}_n||^2\right] \leq \\\frac{1}{K\eta}\sum_{k=1}^K p_k|h_k|^2(||\vtSigma||^2+\mathbb{E}\left[||\nabla F(\vtW_n)||^2\right]) +\frac{d\sigma_z^2}{MK^2\eta}.
\end{split}
\end{equation}
With $\mathbb{E}\left[\Delta\hat{\vtW}_n\right]$ and $\mathbb{E}\left[||\Delta\hat{\vtW}_n||^2\right]$ evaluated in \eqref{eq:expected_value} and \eqref{eq:expected_square}, we go back to the model distance. Taking the expectation on both sides of \eqref{eq:model_distance2} yields
\begin{equation}
\label{eq:model_distance3}
\begin{split}
    \mathbb{E}[r_{n+1}^2] \leq \mathbb{E}[r_n^2] - \sum_{k=1}^K\frac{2\beta\sqrt{p_k}|h_k|\mathbb{E}[\nabla F(\vtW_n)^T(\vtW_n - \vtW^*)]}{K\sqrt{\eta}}\\
    +\frac{\beta^2}{K\eta}\sum_{k=1}^Kp_k|h_k|^2(||\vtSigma||^2 + \mathbb{E}[||\nabla F(\vtW_n)||^2]) + \frac{d\sigma_z^2\beta^2}{MK^2\eta}.
\end{split}
\end{equation}
Now we are ready to introduce the \ac{FL} loss by utilizing strong convexity and Lipschitz smoothness. We do this by rewriting Lemma \ref{lemma:stronglyconvex_smooth} to
\begin{equation}
\label{eq:stronglyconvex_smooth2}
\begin{split}
    &\mathbb{E}\left[\nabla F(\vtW_n)^T(\vtW_n-\vtW^*)\right] \geq \\&\frac{\mu L}{\mu+L}\mathbb{E}[r_n^2]+\frac{1}{\mu+L}\mathbb{E}\left[||\nabla F(\vtW_n)||^2\right],
\end{split}
\end{equation}
where we have utilized $\nabla F(\vtW^*) = 0$ for the final term on the RHS. Combining \eqref{eq:model_distance3} and \eqref{eq:stronglyconvex_smooth2} yields
\begin{equation}
\label{eq:model_distance4}
\begin{split}
    \mathbb{E}[r_{n+1}^2] \leq \mathbb{E}[r_n^2] - \frac{2\beta}{K\sqrt{\eta}}\sum_{k=1}^K\sqrt{p_k}|h_k|\bigg(\frac{\mu L}{\mu+L}\mathbb{E}[r_n^2] \\+\frac{1}{\mu+L}\mathbb{E}\left[||\nabla F(\vtW_n)||^2\right]\bigg)\\
    +\frac{\beta^2}{K\eta}\sum_{k=1}^Kp_k|h_k|^2(||\vtSigma||^2 + \mathbb{E}[||\nabla F(\vtW_n)||^2]) + \frac{d\sigma_z^2\beta^2}{MK^2\eta}.
\end{split}
\end{equation}
Since this expression is getting long, we use three constants $c_2$, $c_3$, and $c_4$ to simplify it. The model distance is then
\begin{equation}
\label{eq:bound_2}
    \mathbb{E}[r_{n+1}^2] \leq c_2\mathbb{E}[r_n^2] + c_3 + c_4\mathbb{E}[||\nabla F(\vtW_n)||^2],
\end{equation}
where $c_2$ and $c_3$ were defined in \eqref{eq:convergence_constant} and \eqref{eq:c3_constant} respectively. Because of our choice of learning rate in Assumption \ref{as:learning_rate}, we have the following inequality for $c_4$
\begin{equation}
    c_4 \coloneqq \frac{\beta^2}{K\eta}\sum_{k=1}^Kp_k|h_k|^2 - \frac{2\beta}{K\sqrt{\eta}(\mu+L)}\sum_{k=1}^K\sqrt{p_k}|h_k| < 0.
\end{equation}
Since $c_4$ is less than zero, we can rewrite our bound in \eqref{eq:bound_2} as
\begin{equation}
\label{eq:iteration1}
    \mathbb{E}[r_{n+1}^2] \leq c_2\mathbb{E}[r_n^2] + \beta^2c_3.
\end{equation}
At this point, the bound is almost complete. The only thing that remains is to find an inequality comparing $\mathbb{E}[r_{n}^2]$ and $\mathbb{E}[r_{0}^2]$ instead of comparing two adjacent communication rounds. As such, we reduce the iteration counter by one, which yields
\begin{equation}
\label{eq:iteration2}
    \mathbb{E}[r_{n}^2] \leq c_2\mathbb{E}[r_{n-1}^2] + \beta^2c_3.
\end{equation}
We then combine \eqref{eq:iteration1} and \eqref{eq:iteration2} to get
\begin{equation}
\begin{split}
    \mathbb{E}[r_{n+1}^2] \leq c_2^2\mathbb{E}[r_{n-1}^2] + (c_2+1)\beta^2c_3.
\end{split}
\end{equation}
By induction we have
\begin{equation}
    \mathbb{E}[r_{n}^2] \leq c_2^n\mathbb{E}[r_0^2] + \beta^2c_3\sum_{i=0}^{n-1}c_2^i.
\end{equation}
Then we apply $\sum_{i=0}^{n-1}c_2^i < \sum_{i=0}^{\infty}c_2^i = 1/(1-c_2)$ to achieve
\begin{equation}
\label{eq:bound3}
    \mathbb{E}[r_{n}^2] \leq c_2^n\mathbb{E}[r_0^2] + \frac{\beta^2}{(1-c_2)}c_3.
\end{equation}
Finally, we utilize convexity and Lipschitz smoothness from \eqref{eq:convex_smooth} to relate the LHS of \eqref{eq:bound3} to the \ac{FL} loss, which yields
\begin{equation}
\begin{split}
    \mathbb{E}\left[F(\vtW_n)\right] - F(\vtW^*) \leq \frac{L}{2}c_2^n\mathbb{E}[r_0^2] + \frac{\beta^2L}{2(1-c_2)}c_3,
\end{split}
\end{equation}
which is the bound from Proposition \ref{proposition:stronglyconvex}. \qed
\begin{remark}
The $c_2$ variable determines the convergence rate of the algorithm and must fulfill $c_2 < 1$ to guarantee convergence in expectation. Given that the static learning rate $\beta$ follows Assumption \ref{as:learning_rate}, this inequality will always hold. It is easily proven by substituting the $\beta$ from \eqref{eq:learning_rate} into
\begin{equation}
    c_2 \coloneqq 1-\frac{2\beta}{K\sqrt{\eta}}\frac{\mu L}{\mu + L}\sum_{k=1}^K\sqrt{p_k}|h_k|.
\end{equation}
\end{remark}

\section{Proof of Proposition \ref{proposition:convex}}
\label{appendix:convex}
Just as in the first proof, we utilize the properties of convexity and Lipschitz smoothness to relate distance between the optimal global model $\vtW^*$ and the current global model $\vtW_n$ to the \ac{FL} loss function. In contrast to the first proof, we use these properties immediately. Specifically, we start with Lemma \ref{lemma:convex_smooth} and take the expectation on both sides to get
\begin{equation}
\begin{split}
    \mathbb{E}\left[F(\vtW_{n+1})\right] \leq \mathbb{E}\left[F(\vtW_{n})\right] + \mathbb{E}\left[\nabla F(\vtW_{n})^T(\vtW_{n+1}-\vtW_{n})\right] \\
    + \frac{L}{2}\mathbb{E}\left[||\vtW_{n+1}-\vtW_{n}||^2\right].
\end{split}
\end{equation}
Then, we add the global update via \eqref{eq:update_rule_ota} to get
\begin{equation}
\begin{split}
    \mathbb{E}\left[F(\vtW_{n+1})\right] \leq \mathbb{E}\left[F(\vtW_{n})\right] - \beta\mathbb{E}\left[\nabla F(\vtW_{n})^T\Delta\hat{\vtW}_n\right] \\
    + \frac{L\beta^2}{2}\mathbb{E}\left[||\Delta\hat{\vtW}_n||^2\right].
\end{split}
\end{equation}
The impact of the wireless channel is introduced in \eqref{eq:global_update_proof}. Then, following steps similar to \eqref{eq:expected_value_1}-\eqref{eq:expected_value} gives us
\begin{equation}
\begin{split}
    \mathbb{E}\left[F(\vtW_{n+1})\right] \leq \mathbb{E}\left[F(\vtW_{n})\right] \\
    - \frac{\beta}{K\sqrt{\eta}}\sum_{k=1}^K\sqrt{p_k}|h_k|\mathbb{E}\left[||\nabla F(\vtW_{n})||^2\right] \\
     + \frac{L\beta^2}{2}\mathbb{E}\left[||\Delta\hat{\vtW}_n||^2\right].
\end{split}
\end{equation}
Then, we insert the bound for $\mathbb{E}\left[||\Delta\hat{\vtW}_n||^2\right]$ from \eqref{eq:expected_square} to get
\begin{equation}
\label{eq:bound4}
\begin{split}
    &\mathbb{E}\left[F(\vtW_{n+1})\right] \leq \mathbb{E}\left[F(\vtW_{n})\right] \\
    &- \frac{\beta}{K\sqrt{\eta}}\sum_{k=1}^K\sqrt{p_k}|h_k|\mathbb{E}\left[||\nabla F(\vtW_{n})||^2\right] \\
    &+ \frac{L\beta^2}{2K\eta}\left(\sum_{k=1}^Kp_k|h_k|^2(||\vtSigma||^2+\mathbb{E}[||\nabla F(\vtW_n)||^2]) + \frac{d\sigma_z^2}{MK}\right).
\end{split}
\end{equation}
Next, we recognize $c_3$ from \eqref{eq:c3_constant} and substitute it into the bound
\begin{equation}
\begin{split}
    &\mathbb{E}\left[F(\vtW_{n+1})\right] \leq \mathbb{E}\left[F(\vtW_{n})\right] \\
    &- \frac{\beta}{K\sqrt{\eta}}\sum_{k=1}^K\sqrt{p_k}|h_k|\left(1-\frac{L\beta}{2\sqrt{\eta}}\frac{\sum_{k=1}^Kp_k|h_k|^2}{\sum_{k=1}^K\sqrt{p_k}|h_k|}\right)\\
    &\mathbb{E}\left[||\nabla F(\vtW_{n})||^2\right] + \frac{L\beta^2}{2}c_3.
\end{split}
\end{equation}
Then we do the same for $c_1$
\begin{equation}
\begin{split}
    &\mathbb{E}\left[F(\vtW_{n+1})\right] \leq \mathbb{E}\left[F(\vtW_{n})\right] \\
    &- \frac{\beta c_1}{K}\left(1-\frac{L\beta}{2\sqrt{\eta}}\frac{\sum_{k=1}^Kp_k|h_k|^2}{\sum_{k=1}^K\sqrt{p_k}|h_k|}\right)\mathbb{E}\left[||\nabla F(\vtW_{n})||^2\right]\\
    & + \frac{L\beta^2}{2}c_3.
\end{split}
\end{equation}
This expression can be simplified by using Assumption \ref{as:learning_rate2} to bound the second term on the RHS
\begin{equation}
\label{eq:bound5}
\begin{split}
    &\mathbb{E}\left[F(\vtW_{n+1})\right] \leq \\ &\mathbb{E}\left[F(\vtW_{n})\right] - \frac{\beta c_1}{2K}\mathbb{E}\left[||\nabla F(\vtW_{n})||^2\right]
    + \frac{L\beta^2}{2}c_3.
\end{split}
\end{equation}
Next, we are going to upper bound the first term on the RHS of \eqref{eq:bound5}. We use the following standard property of convexity (see equation 3.2 from \cite{boyd2004convex}):
\begin{equation}
    \mathbb{E}[F(\vtW_n)] \leq F(\vtW^*) + \mathbb{E}[\nabla F(\vtW_n)^T(\vtW_n-\vtW^*)].
\end{equation}
Plug this into \eqref{eq:bound5}
\begin{equation}
\label{eq:bound6}
\begin{split}
    &\mathbb{E}\left[F(\vtW_{n+1})\right] \leq F(\vtW^*) + \mathbb{E}[\nabla F(\vtW_n)^T(\vtW_n-\vtW^*)] \\
    &- \frac{\beta c_1}{2K}\mathbb{E}\left[||\nabla F(\vtW_{n})||^2\right]
    + \frac{L\beta^2}{2}c_3.
\end{split}
\end{equation}
Some tedious algebraic manipulation transforms \eqref{eq:bound6} to
\begin{equation}
\label{eq:bound7}
\begin{split}
    &\mathbb{E}\left[F(\vtW_{n+1})\right] \leq F(\vtW^*) \\
    &+ \frac{K}{2\beta c_1}\mathbb{E}\bigg[||\vtW_n-\vtW^*||^2 -||(\vtW_n-\vtW^*) -\frac{\beta c_1}{K}\nabla F(\vtW_n)||^2\bigg] \\
    &+ \frac{L\beta^2}{2}c_3.
\end{split}
\end{equation}
Let $r_n^2 \coloneqq ||\vtW_n - \vtW^*||^2$ (same as \eqref{eq:model_distance}) and 
\begin{equation}
    \tilde{r}_{n+1}^2 \coloneqq ||\vtW_n - \vtW^* - \frac{\beta c_1}{K}\nabla F(\vtW_n)||^2.
\end{equation}
Then \eqref{eq:bound7} becomes
\begin{equation}
\label{eq:bound8}
\begin{split}
    &\mathbb{E}\left[F(\vtW_{n+1})\right] \leq F(\vtW^*) \\
    &+ \frac{K}{2\beta c_1}\mathbb{E}\left[r_n^2-\tilde{r}_{n+1}^2\right] + \frac{L\beta^2}{2}c_3.
\end{split}
\end{equation}
Now, just like in the previous proof, we want to form a bound with respect to $r_0^2$. However, instead of using induction we use a telescoping sum. To set it up, we start by taking a sum of \eqref{eq:bound8} over $n$ iterations to get
\begin{equation}
\label{eq:telescoping_sum}
\begin{split}
    &\sum_{i=1}^n\mathbb{E}\left[F(\vtW_{i})\right] -nF(\vtW^*) \leq \\
    &+ \frac{K}{2\beta c_1}\sum_{i=1}^n\mathbb{E}\left[r_{i-1}^2-\tilde{r}_{i}^2\right] + \frac{nL\beta^2}{2}c_3.
\end{split}
\end{equation}
The sum $\sum_{i=1}^n\mathbb{E}\left[r_{i-1}^2-\tilde{r}_{i}^2\right]$ can be rewritten as
\begin{equation}
    \sum_{i=1}^n\mathbb{E}\left[r_{i-1}^2-\tilde{r}_{i}^2\right] = \mathbb{E}[r_0^2]-\mathbb{E}[\tilde{r}_n^2] + \sum_{i=0}^{n-2}\mathbb{E}[r_{i+1}^2-\tilde{r}_{i+1}^2]
\end{equation}
and the middle terms $\sum_{i=0}^{n-2}\mathbb{E}[r_{i+1}^2-\tilde{r}_{i+1}^2]$ will be upper bounded to a constant. We develop this bound next. To start, we plug in the definition of $r_{i+1}^2$ and $\tilde{r}_{i+1}^2$ into $\mathbb{E}[r_{i+1}^2-\tilde{r}_{i+1}^2]$:
\begin{equation}
\begin{split}
    &\mathbb{E}\left[r_{i+1}^2-\tilde{r}_{i+1}^2\right] = \\ &\mathbb{E}\bigg[||\vtW_{i+1}-\vtW^*||^2 -||\vtW_i-\vtW^*-\frac{\beta c_1}{K}\nabla F(\vtW_i)||^2\bigg].
\end{split}
\end{equation}
Applying \eqref{eq:update_rule_ota} and doing some algebra yields
\begin{equation}
\begin{split}
    &\mathbb{E}\left[r_{i+1}^2-\tilde{r}_{i+1}^2\right] =\beta\mathbb{E}\left[\beta||\Delta\hat{\vtW}_i||^2-2(\vtW_i-\vtW^*)^T\Delta\hat{\vtW}_i\right] \\
    & +\beta\mathbb{E}\bigg[ \frac{2c_1}{K}(\vtW_i-\vtW^*)^T\nabla F(\vtW_i) -\frac{\beta c_1^2}{K^2}||\nabla F(\vtW_i )||^2\bigg].
\end{split}
\end{equation}
Then we apply \eqref{eq:global_update_proof} and after some algebra we have
\begin{equation}
\begin{split}
    &\mathbb{E}\left[r_{i+1}^2-\tilde{r}_{i+1}^2\right]\\ &=\beta^2\mathbb{E}\left[||\Delta\hat{\vtW}_i||^2-\frac{c_1^2}{K^2}||\nabla F(\vtW_i)||^2\right].
\end{split}
\end{equation}
Next, we insert $||\Delta\hat{\vtW}_i||^2$ from \eqref{eq:expected_square} and do some algebra which yields
\begin{equation}
\begin{split}
    &\mathbb{E}\left[r_{i+1}^2-\tilde{r}_{i+1}^2\right]=\beta^2\mathbb{E}\bigg[\frac{1}{K^2\eta}\left(\sum_{k=1}^K\sqrt{p_k}|h_k|\right)^2||\Delta \vtW_{n,k}||^2 \\
    &+ \frac{d\sigma_z^2}{MK^2\eta} - \frac{c_1^2}{K^2}||\nabla F(\vtW_i)||^2\bigg].
\end{split}
\end{equation}
Assumption \ref{as:variance_bound} together with some algebra gives us
\begin{equation}
\begin{split}
    &\mathbb{E}\left[r_{i+1}^2-\tilde{r}_{i+1}^2\right] \leq \beta^2\frac{c_1^2}{K^2}||\vtSigma||^2+\beta^2\frac{d\sigma_z^2}{MK^2\eta}.
\end{split}
\end{equation}
Finally, we apply the Cauchy-Schwarz inequality to get
\begin{equation}
\begin{split}
    &\mathbb{E}\left[r_{i+1}^2-\tilde{r}_{i+1}^2\right] \leq\\ &\beta^2||\vtSigma||^2\frac{\sum_{k=1}^Kp_k|h_k|^2}{K\eta} + \beta^2\frac{d\sigma_z^2}{MK^2\eta} = \beta^2c_3
\end{split}
\end{equation}
That concludes the upper bound on $\mathbb{E}\left[r_{i+1}^2-\tilde{r}_{i+1}^2\right]$ so we plug it back into \eqref{eq:telescoping_sum} to get
\begin{equation}
\label{eq:bound9}
\begin{split}
    &\sum_{i=1}^n\mathbb{E}\left[F(\vtW_{i})\right] -nF(\vtW^*) \leq \\
    &\frac{K}{2\beta c_1}\mathbb{E}\left[r_{0}^2-\tilde{r}_n^2+(n-1)\beta^2c_3\right] + \frac{nL\beta^2}{2}c_3.
\end{split}
\end{equation}
Since $\tilde{r}_n^2$ and $c_3$ are positive, we can add one of each to the expectation in the RHS of \eqref{eq:bound9} without breaking the inequality, which yields
\begin{equation}
\label{eq:bound10}
\begin{split}
    \sum_{i=1}^n\mathbb{E}\left[F(\vtW_{i})\right] -nF(\vtW^*) \leq \frac{K}{2\beta c_1}\mathbb{E}\left[r_{0}^2\right]
    + \frac{n\beta}{2}(\frac{K}{c_1}+L\beta)c_3.
\end{split}
\end{equation}
Finally, we note that $\mathbb{E}\left[F(\vtW_{n})\right] \leq \mathbb{E}\left[F(\vtW_{i})\right]$ for all $i \leq n$. Therefore
\begin{equation}
\begin{split}
    &\mathbb{E}\left[F(\vtW_{n})\right] - F(\vtW^*) \leq \frac{1}{n}\sum_{i=1}^n\mathbb{E}\left[F(\vtW_{i})\right] - F(\vtW^*) \leq\\
    &\frac{K}{2n\beta c_1}\mathbb{E}\left[r_{0}^2\right] + \frac{\beta}{2}\left(\frac{K}{c_1}+L\beta\right)c_3,
\end{split}
\end{equation}
which is the bound from Proposition \ref{proposition:convex}. \qed
\ifCLASSOPTIONcaptionsoff
  \newpage
\fi








\bibliographystyle{IEEEtran}
\bibliography{library}

\end{document}